\documentclass[format=acmsmall, screen=true, review=false, authorversion=true]{acmart}

\acmConference{...}{...}{...}

\usepackage{proof}
\usepackage{graphicx}
\usepackage{xspace}
\usepackage{color} 
\usepackage[all]{xy}
\usepackage{tikz-cd}
\usepackage{hyperref}
\usepackage{mathtools}
\usepackage{mathrsfs}
\usepackage{alltt}
\usepackage{listings}

\usepackage{macros}

\setcopyright{rightsretained}

\begin{document}

\citestyle{acmauthoryear}

\title{Polyregular functions on unordered trees of bounded height}
\author{Miko\l{}aj Boja\'nczyk}
\orcid{0000-0002-7758-1072}
\affiliation{\institution{University of Warsaw} \country{Poland}}
\author{Bartek Klin}
\orcid{0000-0001-5793-7425}
\affiliation{\institution{University of Oxford} \country{UK}}

\begin{abstract}
We consider injective first-order interpretations that input and output trees of bounded height. The corresponding functions have polynomial output size, since a first-order interpretation can use a $k$-tuple of input nodes to represent a single output node. We prove that the equivalence problem for such functions is decidable, i.e.~given two such interpretations, one can decide whether, for every input tree, the two output trees are isomorphic. 

We also give a calculus of typed functions and combinators which derives exactly injective first-order interpretations for unordered trees of bounded height. The calculus is based on a type system, where the type constructors are products, coproducts and a monad of multisets. Thanks to our results about tree-to-tree interpretations, the equivalence problem is decidable for this calculus.

As an application, we show that the equivalence problem is decidable for first-order interpretations between classes of graphs that have bounded tree-depth. In all cases studied in this paper, first-order logic and \mso have the same expressive power, and hence all results apply also to \mso interpretations.
\end{abstract}



\maketitle

\section{Introduction}
This paper is about computational models that transform objects such as strings or trees, are powerful enough to describe interesting programs, but are weak enough to have a decidable halting problem. The point of departure is the class of \emph{polyregular functions}~\cite{polyregular-survey}. This is  a class of string-to-string functions that contains functions such as 
\begin{align*}
\myunderbrace{123 \mapsto 123123123}{squaring} 
\qquad 
\myunderbrace{123 \mapsto 321}{reverse},
\end{align*}
and that can be defined by many equivalent models of computation, including: pebble transducers~\cite[Section 3.1]{DBLP:journals/jcss/MiloSV03}, a certain imperative programming language~\cite[Section 1]{polyregular-survey}, several functional programming languages with the same expressive power~\cite[Sections 3 and 4]{polyregular-survey} and~\cite{polyregular-fold}, and logical interpretations~\cite[Theorem 7]{msoInterpretations}. The general idea is that the polyregular functions are those string-to-string functions which have polynomial output size, and which can be computed by devices similar to automata. Due to their similarity to finite automata, certain problems are known to be decidable for polyregular functions (see e.g.~\cite[Corollary 1.5]{polyregular-survey}). An outstanding open problem about polyregular functions is decidability of the equivalence problem, i.e.~it is not known if one can decide whether two devices (transducers, programs, interpretations) compute the same polyregular function. In this paper we make progress on this question by answering it in a different setup, where instead of strings the underlying data structure is nested multisets. 

Let us begin by explaining the objects studied here. One description (see Section~\ref{sec:multiset-types}) is that these are data types that are constructed from finite sets using the multiset constructor 
\begin{align*}
\multisets X = \text{finite multisets of elements from $X$}.
\end{align*}
(We also allow products and coproducts, but it is only the multiset constructor that can create infinite sets.)  We use the name \emph{multiset types} for such types\footnote{The idea to consider regular functions for multiset types was suggested to us by Marcelo Fiore.}. Another description of multiset types is that they describe trees, since nested multisets can be represented as trees, as explained in the following picture: 
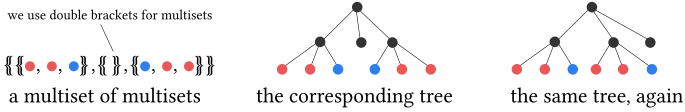
\mypic{26} 
Importantly, the trees are unordered (there is no order on siblings, since multisets are unordered), and their height is bounded (by the nesting depth of the multiset constructor in a multiset type).

This paper is about functions on multiset types or, equivalently, unordered trees of bounded height. The contributions of this paper are:
\begin{enumerate}
    \item In Section~\ref{sec:multiset-types}-\ref{sec:combinators}, we introduce a notion of \emph{polyregular function} for such objects. These are functions between multiset types (or, equivalently, unordered trees of bounded height), that have polynomial output size, and that coincide with regular languages in the case of Boolean outputs. One example of such a polyregular function is the product operation $(\multisets \Sigma) \times (\multisets \Gamma) \to \multisets (\Sigma \times \Gamma)$, that inputs two multisets,  of cardinalities $n$ and $m$, respectively, and outputs a single multiset of cardinality $n \cdot m$. This function has quadratic output size. Another example, which has linear output size, is the multiset union operation of type $\multisets \multisets \Sigma \to \multisets \Sigma$. 
    In our definition of polyregular functions, we propose two models, and prove that they are equivalent. The first model is first-order interpretations, which are a classical model of describing functions using logical formulas~\cite[Chapter 5]{hodges1993model}. The second model uses certain prime functions (such as multiset union) and combinators (such as sequential composition of functions), inspired by a similar system for string-to-string polyregular functions~\cite[Section 3]{polyregular-survey}. The proof that these models describe the same functions on multiset types is inspired by a similar result for strings~\cite[Theorem 7]{msoInterpretations}, except that in the multiset case we can avoid many technicalities that appear in the string case. One of the useful features of multiset types is that first-order (\fo) logic has the same expressive power as monadic second-order (\mso) logic. 
    \item In Sections~\ref{sec:quantifier-free-tree-to-tree}-\ref{sec:equivalence-full},  we prove that the equivalence problem is decidable for the polyregular functions on multiset types. This means that one can decide if two programs represent the same function, i.e.~for every input (which is taken from some multiset type), the two possible outputs are equal. For strings, this problem remains open~\cite[Section 8]{polyregular-survey}, and the decidability result on multisets from this paper is the first significant progress on the problem (Section~\ref{sec:connections-with-strings}, we comment on the connection between the two problems). Our algorithm uses two main ingredients: a quantifier-elimination result, and an algorithm for the quantifier-free case that uses a data structure that involves trees with edge weights taken from a polynomial ring $\mathbb Z[X]$. The equivalence algorithm is the main technical contribution of this paper. 
    \item Finally, in Section~\ref{sec:tree-depth} we illustrate the strength of the equivalence result from the previous item by extending it to graphs of bounded tree-depth.
\end{enumerate}

\paragraph*{Context.}
   One could place the present paper within the wider context of Hilbert's Entscheidungsproblem~\cite{borger2001classical},  which is to decide if a first-order formula is true. This problem can be seen as deciding equality on (the semantics) of formulas, i.e whether a formula $\varphi$ is equal to another formula, such as ``true'' or ``false''. The  problem is famously undecidable, as shown by Church and Turing. One of the most successful methods of recovering decidability has been to use automata methods (tree automata, \mso logic, tree decompositions, etc). The scope 
of such methods is now well understood; they will work if and only if one considers structures that are similar to trees, see~\cite[Theorem 8]{seese1991structure}  and~\cite[Theorem 5.6]{courcelle2007vertex} for formalisations of this statement. 

In this paper, we address a functional version of the  problem: instead of deciding equality on logically defined languages (functions with Boolean outputs), we want to decide it on logically defined functions (which input and output structures). In the functional version, logical formulas are replaced by logical interpretations, and  one asks if for every input, the outputs produced by the two interpretations are isomorphic. This makes the problem nontrivial for  output structures with nontrivial isomorphisms, such as graphs, or even unordered trees.  This difficulty is not present in the language version, because isomorphism creates difficulties only for the output structures. Making progress on this problem seems to be very hard, even for tame structures such as unordered trees of bounded height or graphs of bounded tree-depth, as considered in this paper.

To our best knowledge, the results of this paper about interpretations on bounded tree-depth gives are the first known decidability results about equivalence for graph-to-graph interpretations. Previous results  about equivalence involved trees, with the state of the art being the results on tree-to-tree interpretations from~\cite[Corollary 8.2]{seidlManethKemper2018} and~\cite[Section 3]{boiretReducingTransducerEquivalence2018}. These papers deal with trees of unbounded depth, but the interpretations are restricted to have linear output size. The techniques are based on Hilbert's Basis Theorem and  seem to fail for interpretations of non-linear output size, and also have a certain fragility regarding the type of structures that can be modelled (for example, the results work for linear interpretations on trees with a distinguished root, but the case of trees without distinguished roots remains open~\cite[p.~7]{remarksGraphtoGraph}). The only known decidability results for equivalence of non-linear interpretations is for functions that output numbers~\cite[Corollary 19]{Doueneau-Tabot21}; the techniques used there are based on weighted automata and seem to apply only to outputs that are numbers.

Nested multisets have been studied in database theory, in the context of nested relational algebra~\cite{BUN19953,CLW14,CW19}. Calculi of typed functionals are used there to define transformations between nested collections, similarly to what we aim to do here. It appears that our approach is more restricted than most of the calculi considered there, with the advantage of decidable equivalence. Potential applications of our results in nested relational calculi remain to be investigated.


\section{Multiset types}
\label{sec:multiset-types}
One way of describing trees of bounded height is to view them as values of certain algebraic datatypes. The essential ingredient will be the multiset type constructor, since an unordered tree can be seen as the multiset of its child subtrees. We begin the paper with this approach. Apart from the multiset constructor, we also allow products and coproducts, which will be useful for typing various helper functions. 

\begin{definition}[Multiset types]
 A multiset type is any type that can be generated using the following type constructors:
 \begin{align*}
 \myunderbrace{1}{a set with \\ \scriptsize one element}
 \qquad \qquad \myunderbrace{\Sigma + \Gamma}{disjoint union \\ \scriptsize of sets $\Sigma$ and $\Gamma$}
 \qquad \qquad \myunderbrace{\Sigma \times \Gamma}{product of \\ \scriptsize sets $\Sigma$ and $\Gamma$}
 \qquad \qquad \myunderbrace{\multisets \Sigma}{multisets \\ \scriptsize over set $\Sigma$}.
 \end{align*}
\end{definition}
Strictly speaking, this definition introduces two notions: a syntax of multiset types, and an obvious semantics of each type as a certain set of values. There is no need to be pedantic about this distinction for the moment, but in Section~\ref{sec:logic} we will introduce a more elaborate semantics of multiset types as sets of relational structures over certain vocabularies.

\begin{myexample}\label{ex:unlabelled-trees}
$\multisets^k 1$ is the type of unordered trees of height at most $k$. (The \emph{height} of a tree is defined to be the maximal number of edges on a root-to-leaf path.) The idea is that a tree is represented as the multiset of the representations of its child subtrees. For example, consider $k=2$ and the following tree
 \mypic{6}
 Assuming that we use double set brackets to represent multisets, the above tree is represented as the following multiset of multisets:
 \begin{align*}
 \multiset{\multiset{\ast,\ast}, \emptyset, \emptyset } \in \multisets^2 1,
 \end{align*}
 where $\ast$ denotes the unique element of the set $1$.
\end{myexample}
\begin{myexample} \label{ex:trees} We can extend the trees from the previous example by adding labels from some type $\Sigma$. There are two variants of such trees of height at most $k$:
\begin{itemize}
\item $\edgetrees k \Sigma$: edge-labeled trees, and
\item $ \nodetrees k \Sigma$: node-labeled tress.
\end{itemize}
 These constructors can be simulated using multiset types, by induction on $k$:
 \begin{align*}
 \edgetrees 0 \Sigma & = 1 & \nodetrees 0 \Sigma & = \Sigma \\
 \edgetrees {k+1} \Sigma & = \multisets \left(\Sigma \times \edgetrees k \Sigma\right) 
 & \nodetrees {k+1} \Sigma &= \Sigma \times \left(\multisets \nodetrees k \Sigma\right)
 \end{align*}
 \end{myexample}

We want to study a class of functions between multiset types which we call the \emph{polyregular functions}, meant to be an analogue of the string-to-string polyregular functions from~\cite{bojanczykPolyregularFunctions2018}. There will be two equivalent definitions of this class, one using \fo interpretations (Section~\ref{sec:logic}), and one using combinators (Section~\ref{sec:combinators}).

 \section{Logic}
\label{sec:logic}
We describe a logical approach to polyregular functions on multiset types. The idea is to view each multiset type as a class of structures, in the sense of model theory, and to use \fo interpretations to transform the structures. \fo interpretations are the usual kind of functions from structures to structures, where the universe of the output structure is defined using tuples of elements in the input structure, and the relations of the output structure are defined using \fo formulas. Precise definitions are given below. 


\paragraph*{\fo interpretations.}
We assume that the reader is familiar with the basic notions of first-order (\fo) logic. We use the following terminology. A \emph{vocabulary} is a finite set of relation names, each one with an associated arity in $\set{0,1,2,\ldots}$. A \emph{structure} over such a vocabulary consists of a set, called its \emph{universe}, and for each relation name in the vocabulary a corresponding relation -- of same arity -- over the universe. We assume that the universe in a structure is nonempty.

When we talk about a \emph{class of structures}, we mean any class that contains structures, over the same vocabulary, which is closed under isomorphism. To transform structures of one class into structures of another class, we use \fo interpretations as defined below. This notion is commonly used in model theory (see e.g.~\cite[Chap.~5.3]{hodges1993model}). For a simple example, see Ex.~\ref{ex:numbers} below.

\begin{definition}\label{def:fo-interpretation}
 Let $\Sigma$ and $\Gamma$ be classes of structures. A function $f : \Sigma \to \Gamma$
is called an \fo interpretation if it can be described as follows: there is a finite set $I$, called the \emph{components} of the interpretation, with each component $i \in I$ associated to a \emph{dimension} $k_i \in \set{0,1,\ldots}$, and for every input structure $A\in\Sigma$ the output structure $f(A)$ is defined by: 
\begin{itemize}
\item {\bf Universe.} For every component $i$ there is an \fo formula $\varphi_i$ over the vocabulary of $\Sigma$ with $k_i$ free variables, such that the universe of $f(A)$ is 
\begin{align*}
\coprod_{i \in I} \setbuild {\bar a \in A^{k_i}}{$A \models \varphi_i(\bar a)$}.
\end{align*}
Elements of the above disjoint union will be written as pairs $(i, \bar a)$.
\item {\bf Relations.} For every relation name $R$, say of arity $n$, in the vocabulary of $\Gamma$ and every tuple of components $i_1,\ldots,i_n \in I$, there is an \fo formula $\varphi$ over the vocabulary of $\Sigma$ with $k_{i_1} + \cdots + k_{i_n}$ free variables, such that 
\begin{align*}
f(A) \models R((i_i,\bar a_1),\ldots, (i_n,\bar a_n)) \qquad \iff \qquad A \models \varphi(\bar a_1 \cdots \bar a_n)
\end{align*}
for all $(i_1, \bar a_1),\ldots,(i_n, \bar a_n)$ in the universe of $f(A)$. 
\end{itemize}
\end{definition}

\begin{myremark}
\label{remark:non-injective}
This definition is what is sometimes called an \emph{injective} \fo interpretation. In model theory one typically uses a more general notion of interpretation, where the universe of the output structure is a quotient
\begin{align*}
 \coprod_{i \in A} \setbuild {\bar a \in A^{k_i}}{$A \models \varphi_i(\bar a)$}_{/\sim_i},
 \end{align*}
 where each $\sim_i$ is some equivalence relation that is defined by a formula with $2k_i$ free variables. The formulas defining the relations of the output structure are required to be invariant under these equivalences. This is a more powerful model, which we call \emph{non-injective} interpretations, and which will not be discussed in this paper.
\end{myremark}

\paragraph*{Multiset types as classes of structures.}
A multiset type can be viewed as a class of structures, as explained in the following inductive definition.

\begin{definition}\label{def:type-constructors-on-structures}
 The type constructors are extended to classes of structures as follows:
 \begin{description}
 \item[$1$:] This is the class of structures over the empty vocabulary, that contains only one structure: a universe with one element.
 \item[${\Sigma \times \Gamma}$:] If $\Sigma$ and $\Gamma$ are classes of structures, then $\Sigma\times\Gamma$ is the class over vocabulary
 \begin{align*}
 voc(\Sigma \times \Gamma) = voc(\Sigma) + voc(\Gamma) + \myunderbrace{R(x)}{one extra unary relation name}
 \end{align*}
(here $+$ means disjoint union of vocabularies) that is defined as follows. A structure in $\Sigma \times \Gamma$ is obtained by taking a structure $A \in \Sigma$ and a structure $B \in \Gamma$, and returning the following pair structure $(A,B)$: extend both $A$ and $B$ to the vocabulary $\voc\Sigma + \voc\Gamma$ by interpreting new relation names as empty sets, then take their disjoint union, and finally interpret $R$ as the set of elements that come from $A$. 
 \item[${\Sigma + \Gamma}$:] If $\Sigma$ and $\Gamma$ are classes of structures, then $\Sigma + \Gamma$ is the class over vocabulary
 \begin{align*}
 voc(\Sigma + \Gamma) = voc(\Sigma) + voc(\Gamma) + \myunderbrace{R}{one extra $0$-ary relation name}
 \end{align*}
 which contains structures that are obtained by either: 
 \begin{enumerate}
  \item taking a structure in $\Sigma$, extending it to $\voc \Sigma+ \voc \Gamma$ and interpreting $R$ as true; or
  \item taking a structure in $\Gamma$, extending it to $\voc \Sigma+ \voc \Gamma$ and interpreting $R$ as false.
 \end{enumerate} 
  \item[$\multisets \Sigma$:] If $\Sigma$ is a class of structures,  then $\voc{\multisets \Sigma}$ arises from $\voc \Sigma$  by replacing every $0$-ary relation by a unary relation, and then adding an extra binary relation $\sim$.
  A structure in $\multisets \Sigma$ is obtained by:
  \begin{enumerate}
  \item taking the disjoint union 
  of some structures $A_1,\ldots,A_n \in \Sigma$,
  \item replacing every $0$-ary predicate $R$ from $\voc\Sigma$ by a unary relation that holds for all elements of those $A_i$ for which $R$ is true,
   \item interpreting $\sim$ as the equivalence relation whose equivalence classes are the universes of the structures $A_1,\ldots,A_n$, 
   \item and adding an extra element (called the \emph{root}) that does not participate in any relations\footnote{The root serves two purposes: (a) it guarantees that the universe is nonempty even if $n=0$; and (b) it can be uniquely identified by a first-order formula whenever choosing a unique element is needed. This will come useful in Example~\ref{ex:choices}.}. 
  \end{enumerate}

 \end{description}
\end{definition}
 
In this way, every multiset type can be seen as a class of structures. Therefore one can use \fo sentences to define subsets of multisets types, and \fo interpretations to define functions between multiset types. 

\begin{myremark}
\label{remark:categories}
\fo interpretations are closed under composition, and the identity is a \fo interpretation. Therefore one can consider the following category: the objects are classes of structures, and the morphisms are \fo interpretations modulo equivalence. (This means that two interpretations describe the same morphism if they are equivalent in the sense that for every input structure, the two output structures produced by the two interpretations are isomorphic.) It is not difficult to prove that $1$, $\Sigma + \Gamma$ and $\Sigma \times \Gamma$ in Definition~\ref{def:type-constructors-on-structures} define respectively the terminal object, coproduct, and product in this category, and that $\multisets$ is a monad. We do not elaborate this categorical perspective in this paper, but it will inform our choice of prime functions and combinators in Definition~\ref{def:combinators}.
\end{myremark} 

\begin{myexample}
\label{ex:mk1}
An unordered tree of height at most $k$ can be seen as a relational structure over a vocabulary of $k$ binary relations $\sim_1,\ldots,\sim_k$, with $\sim_i$ interpreted as relating those nodes that have a common ancestor at depth at least $i$. It is easy to construct mutually inverse interpretations between this representation of trees and the more usual one based on a parent-child relation. By Definition~\ref{def:type-constructors-on-structures}, trees represented in this way are exactly structures of type $\multisets^k 1$.
\end{myexample}

 By Lemma~\ref{lem:embedding} below, all multiset types can be viewed as a special case of trees of bounded height. For such structures first-order logic has the same expressive power as monadic second-order logic, see~\cite[Theorem 1.1]{elberfeld2016first}. Therefore, \mso and \fo logics will define the same kind of interpretations between multiset types. We will therefore simply speak about \emph{interpretations} from now on, without specifying that they are \fo interpretations.

\begin{myexample}
\label{ex:numbers}
 Assume that we model the natural numbers as multisets 
 \begin{align*}
 \Nat \eqdef \multisets 1.
 \end{align*}
 Under this representation, which functions $f : \Nat \to \Nat$ can be obtained as interpretations? One example is all polynomials with non-negative coefficients. For example, the function 
 \begin{align*}
 f(n) = 3n^4 + n^2 + 7
 \end{align*}
 is defined by an interpretation with 3 components of dimension 4, 1 component of dimension 2, and 7 components of dimension 0. In each of these components, the universe formula is ``true''. 
 
 Another example is the function that counts the number of non-repeating $k$-tuples in the input, the output of this function is
 \begin{align*}
 n^{(k)} \eqdef n \cdot (n-1) \cdot (n-2) \cdots (n-k+1).
 \end{align*}
 As a polynomial, this function has some negative coefficients. The corresponding interpretation has 1 component of dimension $k$, with a universe formula that selects non-repeating tuples. As it turns out, non-repeating tuples are essentially the only thing that can be done.
 
 \begin{proposition}
 The following are equivalent for every function $f : \Nat \to \Nat$: 
 \begin{enumerate}
 \item $f$ is an interpretation, assuming the representation $\Nat = \multisets 1$; 
 \item there are non-negative co-efficients $a_0,\ldots,a_k \in \set{0,1,\ldots}$ such that
 \begin{align*}
  f(n) = a_0 \cdot n^{(0)} + a_1 \cdot n^{(1)} + \cdots + a_k \cdot n^{(k)},
 \end{align*}
 holds for all sufficiently large $n$. 
 \end{enumerate}
 \end{proposition}
 \begin{proof}
 The implication $2 \Rightarrow 1$ is easy to show. For the converse, it is enough to prove the claim for interpretations with one component; since having several components corresponds to taking a sum of functions with non-negative co-efficients. Suppose then that $f$ is defined by an interpretation with one component of dimension $k$. The corresponding function inputs a number $n$, and returns the number of $k$-tuples that satisfy some first-order formula $\varphi(x_1,\ldots,x_k)$ in the structure that has $n$ elements and no nontrivial relations. Over such a structure, quantifiers are not useful, at least as long as $n$ exceeds some threshold. Counting tuples that satisfy a quantifier-free formula gives a function as in item (2).
 \end{proof}
 
 As a corollary, e.g.~the decrement function  $ f(n) = \max(n-1,0)$ is not an interpretation.
\end{myexample}

 \section{Derivable functions}
\label{sec:combinators}
In this section, we give an alternative definition of polyregular functions on multiset types, which uses combinators. Then we prove that this definition is equivalent to the interpretations from Section~\ref{sec:logic}. 

The idea is to start with certain prime functions, such as multiset union, and close them under certain combinators, such as composition of functions. A brutal approach would be to start with all functions that are interpretations; in this case, the combinators would not be needed. However, in the presence of combinators, only a small set of prime functions is needed. Two of these functions are explained in the following examples.

\begin{myexample}\label{ex:de-singleton}
For a multiset type $\Sigma$, consider the \emph{de-singleton function} of type 
 \begin{align*}
 \multisets \Sigma \to 1 + \Sigma
 \end{align*}
 that maps singleton multisets to their unique elements, and otherwise returns the unique element of~$1$. (Like almost all prime functions, this is not a single function but a family of functions indexed by multiset types.) We claim that this is a first-order interpretation. This has a component of dimension zero, used to produce the error value in the case of non-singleton inputs. The universe formula for this component says that the input multiset is not a singleton, i.e.~that it either consist solely of the root or it contains some non-root elements which are not related by the equivalence relation $\sim$ from the definition of $\multisets \Sigma$:
 \[
 	\forall x.\neg(x\sim x) \quad\lor\quad \exists xy.(x\sim x \land y\sim y \land \neg(x\sim y))
 \] 
For singleton input multisets, the input structure should simply be copied to the output, removing the root. This is done by using an additional component of dimension one, whose universe formula says that the argument is not the root of the multiset. 
\end{myexample}

\begin{myexample}
 \label{ex:choices} Let $\Sigma$ be a multiset type. One of our prime functions, and the only one whose growth is more than linear, is \emph{choices} of the type 
 \begin{align*}
 \multisets \Sigma \to \multisets(\Sigma \times \multisets \Sigma).
 \end{align*}
This function, given a multiset 
\begin{align*}
 \multiset{A_1,\ldots,A_n} \qquad \text{where $A_i \in \Sigma$},
 \end{align*}
outputs the following multiset with the same number of elements:
\begin{align*}
 \multiset{(A_1,B_{1}),\ldots,(A_n,B_{ n})} \qquad \text{where } B_i = \multiset{A_1,\ldots,A_n} - \multiset {A_i}.
\end{align*}
The idea is that the output contains all possible multisets that can be obtained from the input multiset by choosing some distinguished element.

Let us prove that this function is a first-order interpretation.
To represent an element of $A_i$ in the output structure, we simply use the same element in the input structure. These elements are represented using a component of dimension one. To represent an element of $B_i$ in the output structure, we use a pair $(a,b)$, where $a$ is the {\em root} of $A_i$ (see below) and $b$ is an element of some $A_j$ in the input structure with $j \neq i$. Formally, the universe formula for the corresponding component of dimension two is: \begin{align*}
root_\Sigma(a) \land \neg (a \sim b),
\end{align*}
where $\sim$ is the equivalence relation from the definition of $\multisets \Sigma$, and $root_\Sigma$ is a formula that is satisfied for a unique element of every structure in $\Sigma$. Such a formula is defined by induction on $\Sigma$, using the root in the induction step for multisets.
\end{myexample}

The above examples explain the only two non-obvious prime functions. 
The remaining prime functions and all the combinators are straightforward enough that we simply give their types and names, and we assume that the reader can guess their definitions. 

\begin{definition}\label{def:combinators}
 The \emph{derivable} functions are the least class of functions that:
 
 \begin{itemize}
 \item {\bf Prime functions:} contains the following functions for all types $\Sigma, \Sigma_1,\Sigma_2$: 
\begin{align*}
 \textit{union} \colon& \multisets \multisets \Sigma \to \multisets \Sigma && \text{(multiset union)} \\
 \textit{add}\colon& \Sigma \times \multisets \Sigma \to \multisets \Sigma && \text{(add one element)} \\
 \textit{choices}\colon& \multisets \Sigma \to \multisets( \Sigma \times \multisets \Sigma) && \text{(see Ex.~\ref{ex:choices})} \\ 
 \textit{de-singleton}\colon& \multisets \Sigma \to 1 + \Sigma && \text{(see Ex.~\ref{ex:de-singleton})} \\ 
 \textit{empty}\colon& 1 \to \multisets 1 && \text{(constant empty multiset)} \\ 
 \textit{id}\colon& \Sigma \to \Sigma && \text{(identity)} \\
 \pi_i\colon& \Sigma_1 \times \Sigma_2 \to \Sigma_i&& \text{(projection for $i=1,2$)} \\
 \iota_i\colon& \Sigma_i \to \Sigma_1 + \Sigma_2&& \text{(coprojection for $i=1,2$)} \\ 
 \textit{dist}\colon& \Sigma \times (\Sigma_1 + \Sigma_2) \to \Sigma \times \Sigma_1 + \Sigma \times \Sigma_2&& \text{(distribute $\times$ over $+$)} 
 \end{align*}
 \item {\bf Combinators:} is closed under applying the following combinators:
 \begin{align*}
 \frac{f_1 : \Sigma \to \Gamma_1 \quad f_2 : \Sigma \to \Gamma_2}{\langle f_1, f_2 \rangle : \Sigma \to \Gamma_1 \times \Gamma_2}\ \text{pairing} \qquad 
 \frac{f_1 : \Sigma_1 \to \Gamma \quad f_2 : \Sigma_2 \to \Gamma}{[f_1, f_2] : \Sigma_1 + \Sigma_2 \to \Gamma} \ \text{co-pairing} \\[1ex] 
 \frac{ f : \Sigma \to \Gamma}{\multisets f : \multisets \Sigma \to \multisets \Gamma}\ \text{mapping} \qquad 
 \frac{f : \Sigma_1 \to \Sigma_2 \quad g : \Sigma_2 \to \Sigma_3}{f;g : \Sigma_1 \to \Sigma_3} \ \text{composition}
 \end{align*}
\end{itemize}
\end{definition}

\begin{myremark}
It is easy to prove that all prime functions are interpretations (see Theorem~\ref{thm:interpretations-equal-derivable}).
In fact, recalling categorical considerations from Remark~\ref{remark:categories}, most prime functions are natural transformations between the corresponding functors on the category of \fo interpretations. A notable exception is {\em choices}, whose \fo definition in Example~\ref{ex:choices} relies on {\em root} formulas that are defined only for multiset types and not for arbitrary structures. The need for {\em root}s would disappear if we considered non-injective \fo interpretations as in Remark~\ref{remark:non-injective}; then {\em choices} would become a natural transformation. It therefore seems natural to leave the full categorical framework to future work, until we extend our results to non-injective interpretations. For the present, we can use the flexibility afforded by our elementary framework to derive some other non-natural functions by induction on the structure of multiset types.
\end{myremark}

\begin{myexample} Although the list of prime functions has the identity for every multiset type $\Sigma$, it is only really needed in type $1 \to 1$. For the remaining types, it can be derived using combinators. Also, for every type $\Sigma$ we can derive the unique function $!_\Sigma$ of type $\Sigma \to 1$. The proof is by induction on $\Sigma$, in the induction step for $\Sigma = \multisets \Gamma$ one uses the de-singleton operation.

We can also derive functions of type $1\to\Sigma$ by induction on $\Sigma$. For coproducts, this is done with coprojections; for products, with pairing; for multiset types, with the mapping combinator and the constant empty multiset function.
\end{myexample}
\begin{myexample}
 Define the Boolean type to be 
 \begin{align*}
 \bool \eqdef 1+1.
 \end{align*}
 One can easily derive all Boolean operations, e.g.~conjunction $\bool^2 \to \bool$. In fact, for every finite types $\Sigma$ and $\Gamma$, i.e.~types that are built without the multiset constructor $\multisets$, one can derive every function of type $\Sigma \to \Gamma$. 
\end{myexample}
\begin{myexample}
 \label{ex:emptiness-test}
 The emptiness test 
 \begin{align*}
 empty : \multisets 1 \to \bool,
 \end{align*}
 that returns ``true'' for the empty multiset and ``false'' for nonempty ones, is derived as the composition of the following operations:
  \[
 \xymatrix{
 \multisets 1 
 \ar[rr]^-{\langle {!_{\multisets 1}, id}\rangle}
 &&
 1 \times \multisets 1
 \ar[r]^-{\textit{add}}
 &
 \multisets 1
 \ar[rr]^-{\textit{de-singleton}}
 &&
 1 + 1 = \bool.
 }
 \]
The idea is that the empty multiset is the only one to which we can add and element to obtain a singleton. 
\end{myexample}
\begin{myexample}\label{ex:disjunction}
 We now derive the function 
 \begin{align*}
 \lor : \multisets(\bool) \to \bool
 \end{align*}
 that implements disjunction of unbounded arity: the function checks if the input multiset contains at least one ``true'' value. This function is derived as follows. To the input, which is a multiset of Booleans, we apply $\multisets f$, where $f : \bool \to \multisets 1$ is the function that maps ``false'' to the empty multiset and ``true'' to the singleton multiset. After applying multiset union to the result, we test the resulting value for emptiness.
 
Using the same idea, for any multiset types $\Sigma$ and $\Gamma$ we can derive a function
 \begin{align*}
 \dropp : \multisets(\Sigma+\Gamma) \to \multisets \Sigma
 \end{align*}
 that erases all elements of $\Gamma$ from the given multiset.
\end{myexample}


\begin{myexample}
\label{ex:strength}
For multiset types $\Sigma$ and $\Gamma$, consider the function {\em strength} of type
\[
 \multisets \Sigma\times \Gamma \to \multisets(\Sigma \times \Gamma)
\]
which, given an argument $B\in\Gamma$ and a multiset 
\begin{align*}
 \multiset{A_1,\ldots,A_n} \qquad \text{where $A_i \in \Sigma$},
 \end{align*}
 returns the multiset of pairs:
 \begin{align*}
 \multiset{(A_1,B),\ldots,(A_n,B)}
 \end{align*}
 with the same number of elements.  
 
 Let us prove that strength is derivable. Given a structure in $\multisets \Sigma\times\Gamma$, perform the following steps:
\begin{enumerate}
\item use suitable coprojections to get to \hfill $\multisets (\Sigma+\Gamma) \times (\Sigma+\Gamma)$;
\item use {\em add} to get to \hfill $\multisets (\Sigma+\Gamma)$;
\item use {\em choices} to get to \hfill $\multisets ((\Sigma+\Gamma)\times \multisets(\Sigma+\Gamma))$;
\item use {\em dist} to get to \hfill
$\multisets(\Sigma\times\multisets(\Sigma+\Gamma)+\Gamma\times \multisets(\Sigma+\Gamma))$;
\item use $\dropp$ from Example~\ref{ex:disjunction} twice, to get to \hfill
$\multisets (\Sigma\times \multisets\Gamma)$;
\item de-singleton $\multisets\Gamma$, distribute and use $\dropp$ again to get to \hfill $\multisets (\Sigma\times \Gamma)$.
\end{enumerate}
 \end{myexample}

\paragraph*{Trees of bounded height.}
In Examples~\ref{ex:unlabelled-trees} and~\ref{ex:mk1} we showed how unlabelled trees of height at most $k$ can be modeled as structures of  type $\multisets^k 1$. 
The following lemma shows that such types are general, i.e.~all multiset types can be encoded in them. We use the following notion of encoding: we say that a type $\Sigma$ \emph{encodes} in a type $\Gamma$ if there are derivable functions
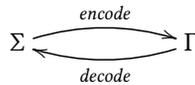
\[\xymatrix{
\Sigma\ar@/^/[rr]^{\textit{encode}} && \Gamma\ar@/^/[ll]^{\textit{decode}} 
}\]
 such that {\em encode;decode} is the identity on $\Sigma$. 
\begin{lemma}\label{lem:embedding}
 Every multiset type encodes into $\multisets^k 1$ for some $k \in \set{0,1,\ldots}$.
\end{lemma}
\begin{proof}
 Induction on the structure of the type. For the type $1$, there is nothing to do. For a type $\multisets \Sigma$, we simply apply the mapping combinator to the encoding for $\Sigma$. We are left with products and coproducts. We only describe the construction for products $\Sigma \times \Gamma$; coproducts can be treated in a similar way. By induction assumption, each of the types $\Sigma$ and $\Gamma$ encodes into trees of some height. We can assume that the heights are the same, because $\multisets^k 1$ encodes into $\multisets^{k+1}1$. We pair the encodings as follows: a pair in $\Sigma \times \Gamma$ is sent to the following tree
 \mypic{4}
 This tree lives in $\multisets^{k+3} 1$, where $k$ is the height of trees needed to encode $\Sigma$ and $\Gamma$.
 It remains to describe the decoding. We explain how to extract the $\Sigma$ component. In the argument below, a \emph{singleton tree} is a tree where the root has a single child. The projection to the $\Sigma$ coordinate is performed in five stages, summarised in the following picture:
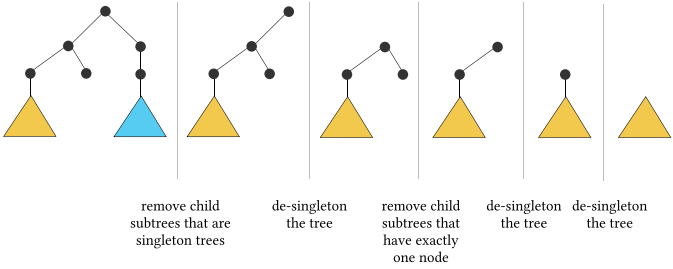
\mypic{5}
We omit the straightforward details of the construction. A similar construction is used for extracting the $\Gamma$ component. 
\end{proof}

\subsection{Completeness of the derivable functions}
\label{sec:completeness}
Let us prove that on multiset types, first-order interpretations coincide with derivable functions.
\begin{theorem}\label{thm:interpretations-equal-derivable}
 Let $f : \Sigma \to \Gamma$ be a function between multiset types. Then $f$ is a first-order interpretation if and only if it is derivable.
\end{theorem}

The easier direction is \emph{soundness}: every derivable function is also an interpretation. This is proved by a straightforward induction on the derivation: all prime functions are easily seen to be first-order interpretations (see Examples~\ref{ex:de-singleton}-\ref{ex:choices}), and first-order interpretations are easily seen to be closed under applying the combinators.

We now move to the harder part of the theorem, which we call \emph{completeness}: every first-order interpretation can be derived. This is proved first for Boolean outputs, and then for general outputs.
\subsubsection{Boolean outputs}
We begin by deriving interpretations with Boolean outputs, which is the same as first-order sentences. (Indeed, all structures of type $\bool$ have a one-element universe with a single nullary relation, so the only ingredient of the interpretation is the sentence which is the interpretation of that relation.) This is done by inlining the semantics of \fo logic into our type system, with the {\em choices} function used to simulate a single quantifier. 

In the proof it will be convenient to think of structures with distinguished elements as being a multiset type. 
This is done using the following type constructor, where $k \in \set{0,1,\ldots}$:
\begin{align*}
\pointed k \Sigma \eqdef \setbuild { (A, \bar a)}{$A \in \Sigma$ and $\bar a$ is a tuple of $k$ distinguished elements}.
\end{align*}
We do not assume that the distinguished elements are pairwise distinct. 
This type constructor can be encoded using the existing type constructors:
\begin{align*}
\pointed 0 \Sigma &\equiv \Sigma &
\pointed 1 1 &\equiv 1 &
\pointed 1 (\Sigma \times \Gamma) &\equiv \pointed 1 \Sigma \times \Gamma + \Sigma\times\pointed 1 \Gamma \\
\pointed {k+1} \Sigma &\equiv V_1(V_k\Sigma) &
\pointed 1 (\Sigma + \Gamma) &\equiv \pointed 1 \Sigma + \pointed 1 \Gamma &
\pointed 1 \multisets \Sigma &\equiv (1+\pointed 1 \Sigma) \times \multisets\Sigma
\end{align*}
(the $1+$ in the multiset case is due to the fact that the root element might be distinguished). As a result, for any multiset type $\Sigma$, we can view $\pointed k \Sigma$ as a multiset type. The semantics of a first-order formula $\varphi$ over the vocabulary of $\Sigma$ that has $k$ free variables is an interpretation of type 
\begin{align*}
\sem \varphi : \pointed k \Sigma \to \bool.
\end{align*}
The following lemma shows that such functions are derivable.

 \begin{lemma}\label{lem:derive-fo-boolean}
 Let $\Sigma$ be a multiset type. For every first-order formula $\varphi$ over the vocabulary of $\Sigma$ with $k$ free variables, the function $\sem{\varphi}$ is derivable. 
 \end{lemma}
 \begin{proof}
 Induction on the size of $\varphi$. The induction base, when $\varphi$ is an atomic formula, is by straightforward induction on $\Sigma$. For example, the equality predicate $x=y$ is modeled as a function
 \[
 	eq_\Sigma : \pointed 2 \Sigma \to \bool,
 \]
easily derived by induction on $\Sigma$. 
 
For the induction step of negation $\neg \varphi$, we simply compose the function of $\varphi$ with the negation operation of type $\bool \to \bool$. The same works for disjunction and conjunction, except that we use pairing. 
 
 The last remaining induction step is that of a quantified formula of the form $\exists x \varphi$ or $\forall x \varphi$, where $\varphi$ is a formula with $k+1$ free variables. Consider the \emph{extension function} of type:
 \begin{align*}
\ext_\Sigma: \pointed k \Sigma \to \multisets \pointed {k+1} \Sigma
 \end{align*}
 that extends the input valuation in all possible ways by adding one element to the tuple. Since by definition $V_{k+1}\Sigma=V_1(V_k\Sigma)$, it is enough to derive this function for $k=0$. This is done by induction on the type $\Sigma$; for $\ext_{\multisets \Sigma}$ we use choices and {\em strength} from Example~\ref{ex:strength},
and for $\ext_{\Sigma\times\Gamma}$ we use {\em strength} again.
 
  The functions $\sem{\exists x \varphi}$ and $\sem{\forall x \varphi}$ are derived by first applying the extension function, then applying  $\sem\varphi$ to each element of the resulting multiset, and then applying disjunction/conjunction from Example~\ref{ex:disjunction} to the resulting multiset of Booleans.
 \end{proof}

\subsubsection{General outputs}
\label{sec:completeness-general}
We now show how to derive first-order interpretations that have outputs which can be any multiset type, not just the Booleans. Thanks to Lemma~\ref{lem:embedding}, which says that every multiset type encodes into trees, the problem will reduce to functions with tree outputs.

\paragraph*{Valuation trees.} To deal with tree outputs, we introduce a data structure called valuation trees.
For a structure $A$ and a number $k \in \set{0,1,\ldots,}$, the \emph{valuation tree of height $k$ for $A$} is defined as follows. Nodes of the tree are pairs $(A, \bar a)$ where $\bar a$ is a list of at most $k$ distinguished elements in $A$; in particular, the structure $A$ is copied in each node in the tree. The ancestor relation is the prefix relation on the tuples of distinguished elements; in particular, the root of the tree uses the empty tuple of distinguished elements. The height of this tree is $k$, and the leaves are labelled by tuples of $k$ distinguished elements. 

For a multiset type $\Sigma$, the valuation trees of height $k$ for structures from $\Sigma$ can be seen as a type: 
\begin{align*}
\nodetrees k (\pointed {\leq k} \Sigma), \quad \text{ where } \quad \pointed {\leq k}\Sigma \quad \text{ is defined as } \quad \pointed 0 \Sigma + \cdots + \pointed k \Sigma\end{align*}
and $\nodetrees k$ is the constructor of node-labelled trees from Example~\ref{ex:trees}.
Not all elements of this type are valuation trees, since a valuation tree requires that the label of a parent node is obtained from the label of any child node by removing the last distinguished element. For every $k$ and multiset type $\Sigma$, there is a derivable function:
\begin{align*}
\valtree_\Sigma^k: \Sigma \to \nodetrees k (\pointed {\leq k} \Sigma)
\end{align*}
that maps a structure to its valuation tree. This is by induction on $k$, with $\valtree_\Sigma^{k+1}$ derived by:
\[\xymatrix{
\Sigma
 \ar[r]^-{\langle id,\ext_\Sigma\rangle} &
\Sigma\times \multisets \pointed 1 \Sigma
 \ar[rr]^-{id\times\multisets(\valtree_{\pointed 1 \Sigma}^k)}
 &&
\Sigma\times \multisets  \nodetrees k {\pointed {\leq k} {\pointed 1 \Sigma}}
\ar[r]
&
 \nodetrees {k+1} {\pointed {\leq {k+1}} \Sigma}
 } \]
 with the function on the right derived from obvious coprojections of $\Sigma=\pointed 0\Sigma$ and $\pointed {\leq k}\pointed 1\Sigma$ into $\pointed {\leq {k+1}}\Sigma$. The extension function $\ext$ was derived in the proof of Lemma~\ref{lem:derive-fo-boolean}.

\paragraph*{Completeness proof.}
Using valuation trees, we finish the completeness proof for Theorem~\ref{thm:interpretations-equal-derivable}. Suppose that we want to derive a first-order interpretation 
$f : \Sigma \to \Gamma$.
Apply Lemma~\ref{lem:embedding} to the output type, yielding derivable functions
\[\xymatrix{
 \Gamma 
 \ar@/^/[rr]^{\textit{encode}}
 &&
 \multisets^n 1.
 \ar@/^/[ll]^{\textit{decode}}
}\]
Since the decoding function is derivable, it is enough to show that the composition
\[\xymatrix{
\Sigma 
\ar[rr]^-{f;\textit{encode}}
&&
\multisets^n 1 = \nodetrees n 1
}\]
 is derivable. This function has tree outputs of bounded height, and it is a first-order interpretation. 
 
 For every such interpretation, there is an equivalent interpretation such that: 
 \begin{itemize}
 \item[(a)] for every dimension $i$ there is at most one component of this dimension; and 
 \item[(b)] the ancestor ordering on nodes of the output tree is the prefix relation on tuples. 
 \end{itemize} 
To ensure (b) only, build an interpretation where nodes of the output structure are sequences (paths in the tree) of output nodes in the original interpretation, with new components arising from sequences (of length up to $n$) of the original components. The original parent-child relation now becomes part of the universe condition, and the new ancestor relation is just the prefix order. To furthermore ensure (a) in this construction, extend the dimension of every component (arising from a sequence of the original components) by adding a certain number of dummy coordinates at the end (and also, necessarily, in the middle, to ensure that the parent-child relation is still a prefix relation). 
 
Assuming that $f;encode$ has properties (a) and (b), and the maximal dimension used in the interpretation is $k$ (typically $k$ will be larger than $n$), one can decompose $f;encode$ into four steps:
\begin{enumerate}
 \item Compute the valuation tree of height $k$, yielding a result in type $ \nodetrees k \pointed {\leq k}\Sigma$.
 \item Consider the universe formula of the first order-interpretation defining $f;encode$, viewed as a function of type
$\pointed {\leq k}\Sigma \to \bool$.
 Apply this function to the label of each node in the result of the previous step, 
yielding a result in type $\nodetrees k \bool$.
\item In the tree computed so far, keep only nodes with label ``true'', while preserving the ancestor relation, yielding a tree in $\nodetrees k 1$.
\item By the semantic properties of the first-order interpretation, we know that the output from the previous step will have height $n$. However, we need to explicitly cast it into the output type $\nodetrees n 1$. We therefore apply a derivable function that keeps only nodes at distance at most $n$ from the root, yielding a tree in the correct output type $\nodetrees n 1$.
\end{enumerate}
All these four steps are derivable. Indeed, as we have already remarked, the valuation tree can be computed using a derivable function. The second step is derivable by Lemma~\ref{lem:derive-fo-boolean}. Deriving the last two steps is straightforward.

\section{Deciding equivalence for quantifier-free interpretations on trees of bounded height}
\label{sec:quantifier-free-tree-to-tree}

In the previous section we presented a list of prime functions and combinators that together generate all interpretations  that input and output trees of bounded height. It is clear that these functions and combinators satisfy various equalities: pairing of the two projections from a product type is equal to the identity, the mapping combinator preserves function composition, and so on. One wonders whether there is an equational axiomatisation of the equivalence of  interpretations presented in this way. Ideally, that equational theory would be decidable. We do not attempt to provide such a theory in this paper. However, we will prove that equivalence of interpretations is indeed decidable. We shall work directly with the logical presentation of interpretations, without a direct reference to the prime functions and combinators.

The problem is to decide, given two \fo interpretations on multiset types, if they produce isomorphic outputs for every input. We first solve a special case of this problem, where the input and output types represent trees of bounded height, and -- more importantly -- the functions are quantifier-free interpretations. In Section~\ref{sec:equivalence-full} we will reduce the general problem to this special case. 

We begin by specifying the exact variant of trees that is used, and their representations as logical structures. We need to be quite specific here, since for quantifier-free interpretations the exact choice of vocabulary plays a role. In Section~\ref{sec:equivalence-full} these details will become unimportant, as quantified formulas allow easy translations between various representations of trees.

For the output trees, we consider unlabelled trees with a height bound $n$, where a tree is represented as a set of nodes with the parent-child relation. We will denote the class of such trees by $\multisets^n 1$. This is a slight abuse of notation: in Section~\ref{sec:logic} this indeed denotes the same class of trees, but under a different relational representation.

For the input, we use trees with edge labels from a finite set $\Sigma$ and with a height bound $k$. This corresponds to the type $\edgetrees k \Sigma$ described in Example~\ref{ex:trees}, but again, abusing the notation, we use a representation different from the one arising from Definition~\ref{def:type-constructors-on-structures}. In this \emph{edge represenation}:
\begin{itemize} 
\item the universe of the tree is the set of its edges;
\item for every label $a$, there is a unary relation that selects edges with label $a$;
\item there is a parent function, which maps an edge to its parent edge. For edges without parents, i.e.~edges that originate in the root of the tree, the parent function loops.
\end{itemize}
A useful property of this representation is that induced substructures have an intuitive meaning: they correspond to ``pruned'' trees that arise by removing some subtrees. The empty structure also has a meaning: it represents a root-only tree. This will be convenient when we apply generic results from Section~\ref{sec:patterns} to labelled trees.

We will consider functions of type 
\begin{align*}
\edgetrees k \Sigma \to \multisets^n 1,
\end{align*}
for a finite set $\Sigma$ of edge labels. The functions will be quantifier-free interpretations, assuming tree representations as above. This is the special case of interpretations as in Definition~\ref{def:fo-interpretation}, in which all formulas are quantifier-free. (Note that Definition~\ref{def:fo-interpretation} makes sense in the presence of functions in the input vocabulary.)

This section is devoted to proving the following theorem.

\begin{theorem}\label{thm:decidable-equivalence-qf-trees}
 The following problem is decidable:
\begin{description}
 \item[Instance.] Two quantifier-free interpretations $
f,g:
 \edgetrees k \Sigma \to
\multisets^n1.
 $
 \item[Question.] Are the functions equivalent in the following sense: for every input, the two outputs are isomorphic as trees?
\end{description}
\end{theorem}


%

We will even show that the problem is in polynomial time, assuming a suitable representation of quantifier-free interpretations, which we shall call patterns.

\subsection{Patterns}
\label{sec:patterns}
In our solution to the equivalence problem, we use a more combinatorial representation of quantifier-free interpretations, which is defined by counting embeddings.




Recall that an \emph{embedding} between structures over the same vocabulary is an injective map from the universe of the source structure into the universe of the target structure, that preserves and reflects all relations. 
We will use embeddings mainly for trees modelled using the functional edge representation as described above.
Under that representation, an embedding of one tree into another is an injective map from edges in the source tree to edges in the target tree, that preserves the labels and the parent function. Since the parent function loops for edges without parents, such edges must be mapped to edges without parents. In other words, the root node must be mapped to the root node. Here is a picture:
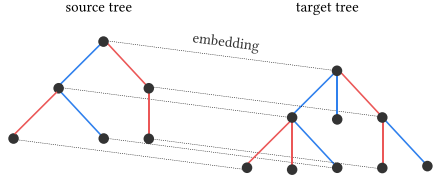
\mypic{9}

The following definition is the key notion of this section. We will apply it to labeled trees under the edge representation, but  the definition makes sense for any class of structures that is closed under induced substructures. Here an induced substructure is obtained by restricting the universe to some (possibly empty) subset that is closed under the functions in the structure. 

\begin{definition}[Pattern]\label{def:pattern}
 Let $\structclass$ be a class of structures that is closed under induced substructures. 
 A \emph{pattern for $\structclass$} is a finite tree $\Phi$ where every node $x$ is labelled by a structure $\Phi(x)$ in $\structclass$, such that:
\begin{itemize}
\item the root node is labelled by the empty structure, and
\item if a node $x$ is the parent of a node $y$, then $\Phi(x)$ is an induced substructure of $\Phi(y)$. 
\end{itemize} 
 \end{definition}

\begin{myexample}
 Here is a picture of a pattern, where the class $\structclass$ is trees of height at most one with edges labelled red or blue:
 \mypic{19}
 The dotted lines in the picture represent the inclusions of induced substructures between a node of the pattern and its children. 
\end{myexample}

Two patterns are isomorphic if there is: (i) a bijection between the underlying trees of the two patterns, that preserves the tree structure; and (ii) a family of isomorphisms, one for each node $x$ in the first pattern, that isomorphically maps the structure labelling $x$ in the first pattern to the structure labelling the corresponding node in the second pattern. The family of isomorphisms in (ii) must be consistent with the inclusion of labels, i.e.~the isomorphism for a node $x$ must extend the isomorphism for its parent.
 
Every pattern $\Phi$ determines a function which inputs structures from $\structclass$ and outputs unlabelled trees, and is defined as follows. For an input structure $A \in \structclass$, the nodes of the output tree are pairs $(x,\alpha)$ where $x$ is a node of the pattern and 
\begin{align*}
\alpha : \Phi(x) \hookrightarrow A
\end{align*}
is an embedding from the structure that labels node $x$ in the pattern to the input structure $A$. The children of a node $(x,\alpha)$ are nodes $(y,\beta)$ such that $y$ is a child of $x$ and the embedding $\beta$ extends $\alpha$. 

By definition, the height of the output tree is bounded by the height of the pattern. Moreover, the function defined by a pattern $\Phi$ is easily seen to be a quantifier-free interpretation. Indeed, one can take a component for each node $x$ in $\Phi$, with dimension equal to the size of the structure $\Phi(x)$, and the universe formula determined by the structure $\Phi(x)$ itself. The parent-child relation is then defined by a suitable sub-tuple formula. All these formulas are quantifier-free.

The following lemma shows that every function of this kind arises from a pattern. 

\begin{lemma}\label{lem:patterns-complete}
 Let $\structclass$ be a class of structures that is closed under induced substructures. Then a function from $\structclass$ to (unlabelled) trees is defined by a pattern if and only if it
is a quantifier-free interpretation that produces outputs of bounded height.
\end{lemma}
\begin{proof}
Consider a quantifier-free interpretation $f$ that, given an input structure in $\structclass$, produces a tree of height at most $n$. As we observed in the completeness proof in Section~\ref{sec:completeness-general}, there exists an equivalent interpretation $\hat{f}$ whose components are sequences of components of $f$ of length up to $n$; the dimension of each such component is the sum of dimensions of the original components in the sequence. The universe formula for such a component combines the universe formulas for the original components with the parent relation formula of the output tree of $f$; this is quantifier-free if all the constituent formulas are quantifier-free. Finally, the parent relation is simply the prefix relation of a suitable length; again, this is described by a quantifier-free formula. 
 
Now, $\hat{f}$ is an interpretation of the same type as $f$, but with the parent relation in the output tree defined as a prefix formula. Every quantifier-free interpretation with this property arises from a pattern. Indeed, the universe formula for each component is (equivalent to) a disjunction of conjunctive formulas that fully describe relations between a finite set of variables. Such a conjunctive formula is essentially a structure from $\structclass$ together with a valuation of the variables that generates the entire structure. A pattern whose nodes are the conjunctive formulas present in $\hat{f}$, with the parent relation inherited from $\hat{f}$, defines the interpretation $\hat{f}$, which is equivalent to $f$.
\end{proof}

Importantly, the above proof is constructive: a pattern that defines a quantifier-free interpretation to unlabelled trees can be effectively reconstructed from (the number $n$ and) a logical description of the interpretation.

\begin{myexample}
Let $\structclass$ be the class of trees of height at most two with edges labelled red or blue, and let $f$ be the function that simply returns the input tree, forgetting the colors. This is a quantifier-free interpretation, defined by a single component of dimension one, with the universe formula ``true'' and an evident parent relation. It is defined by the pattern:
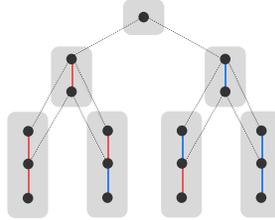
\mypic{24}
In particular, the root node of this pattern has exactly one embedding into any input tree, and that embedding is the root of the output tree.
\end{myexample}

\subsection{Patterns with tree inputs}
Although patterns make sense for general classes of input structures, we focus on patterns where the input class is trees of bounded height labelled by some finite alphabet, as in Theorem~\ref{thm:decidable-equivalence-qf-trees}. By Lemma~\ref{lem:patterns-complete}, 
all functions from the decision problem in Theorem~\ref{thm:decidable-equivalence-qf-trees} can be described using patterns, and these patterns can be computed from the original representation in terms of quantifier-free formulas. Therefore, Theorem~\ref{thm:decidable-equivalence-qf-trees} boils down to deciding equivalence for tree-to-tree functions defined by patterns. We will show that the latter problem is decidable, even in polynomial time, for a very simple reason: different patterns define different tree-to-tree functions. 

\begin{theorem}\label{thm:patterns-must-be-isomorphic}
 Consider two patterns $\Phi_1$ and $\Phi_2$ that define tree-to-tree functions. If the patterns are non-isomorphic, then the corresponding tree-to-tree functions are non-equivalent, i.e.~for some input tree, the two output trees are non-isomorphic.
\end{theorem}

From this we can immediately conclude Theorem~\ref{thm:decidable-equivalence-qf-trees}. Indeed, consider two tree-to-tree functions $f_1$ and $f_2$, as in Theorem~\ref{thm:decidable-equivalence-qf-trees}. By Lemma~\ref{lem:patterns-complete}, these functions are defined by patterns, say $\Phi_1$ and $\Phi_2$. By Theorem~\ref{thm:patterns-must-be-isomorphic}, the functions are equivalent if and only if the patterns are isomorphic. Isomorphism of patterns can be decided in polynomial time using a straightforward dynamic programming algorithm. Note, however, that there is an exponential overhead in translating a quantifier-free interpretation $f_i$ to a pattern $\Phi_i$, hence the entire algorithm is not polynomial time.

The rest of Section~\ref{sec:quantifier-free-tree-to-tree} is devoted to proving Theorem~\ref{thm:patterns-must-be-isomorphic}. The general idea is that any two non-isomorphic patterns are distinguished by an input tree where every node, apart from leaves, has a large number of children of every color. 
%
%
%
%
%
%
To generate a distinguishing input, first a finite labeled tree $t$ of the appropriate depth will be chosen, for example:

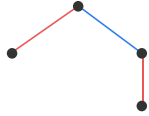
\mypic{16}
Then every edge in $t$ will be cloned into several copies:

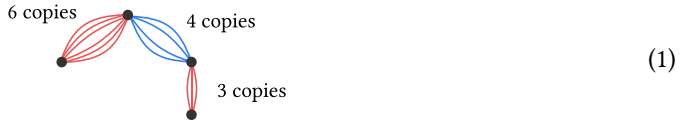
\mylabelledpic{17}{pic:thick-tree}
Finally, the resulting graph will be unfolded into a tree, with the result being:
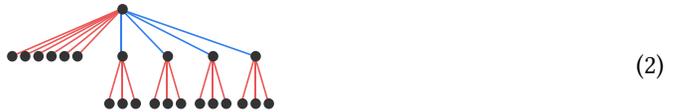
\mylabelledpic{18}{pic:unfolded-tree}

%


We will show that, for an appropriate $t$, cloning every edge into sufficiently many copies will be enough to distinguish any two given patterns. Actually $t$ will not be particularly complicated, it will simply be a large complete tree of the appropriate depth.

To make these ideas precise, we shall consider trees edge-labeled with algebraic expressions that will denote the varying number of copies for every edge.

\subsection{Symbolic trees}
\label{sec:symbolic-trees}
Our proof will rely on a representation of sets of input trees, which is called \emph{symbolic trees}. Intuitively speaking, a symbolic tree is a tree, where each edge is labelled by a pair consisting of a colour from a finite set and a multivariate polynomial with integer coefficients. In the following picture the colours are red and blue, and the polynomials use variables $x$ and $y$:
\mypic{12}
The formal definition of symbolic trees, given below, defines them as the least set closed under taking linear combinations of pairs (colour, smaller symbolic tree). One should think of such a linear combination as representing a tree where the elements of the linear combination are child subtrees, and the coefficients describe the polynomials; in particular, the zero linear combination represents a symbolic tree with a root node and no children.

\begin{definition}\label{def:symbolic-tree}
 For finite sets $\Sigma$ and $X$, the set of \emph{symbolic trees with edge colours in $\Sigma$ and variables $X$} is the least solution $T$ to the inequality
 \begin{align*}
T \supseteq \text{finite formal linear combinations of $\Sigma \times T$ with coefficients in $\Int[X]$}.
\end{align*}
\end{definition}

As mentioned before, a symbolic tree can be represented by a tree with edges labelled by pairs (colour, polynomial). This representation is not unique. For example, here are two representations of the same symbolic tree: 
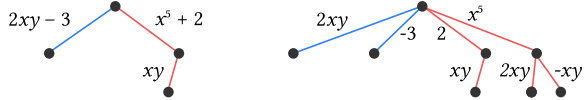
\mypic{20}

The representation does become unique if we require that one cannot have two edges that have the same source, same colour (but possibly different polynomials), and isomorphic subtrees.
Symbolic trees are similar to trees modulo bisimulation; the latter would correspond to Definition~\ref{def:symbolic-tree} if we used the Boolean semiring instead of the ring of polynomials $\Int[X]$.

A \emph{weighted tree} is the special case of a symbolic tree where the set of variables $X$ is empty. This is a tree with possibly negative integer weights on edges. A weighted tree is called \emph{positive} if all of its weights are positive. Such a tree is the same as an isomorphism class of usual, non-weighted, trees. 

The purpose of symbolic trees is to generate weighted trees. This is done by substituting integer values for the variables. If $t$ is a symbolic tree and $\bar a \in \Int^X$ is a choice of integer parameters for its variables such that the value of every polynomial in $t$ is positive, then we write $t(\bar a)$ for the tree that arises by substituting the parameters in each polynomial, and then unfolding the resulting weighted tree as illustrated in~\eqref{pic:unfolded-tree}.

For our proof, we only care about positive weighted trees. One way to generate such trees is to use positive parameters and symbolic trees where all polynomials use only positive coefficients. However, this restriction would be too strong. For example, a typical polynomial that will arise in our proof will be
$x^{(k)}$ (see Example~\ref{ex:numbers}),
which counts non-repeating $k$-tuples. Although this polynomial has negative coefficients, it is \emph{ultimately positive}: there is some $N$ such that if all parameters are larger than $N$ then the value of the polynomial is positive. This is the same as saying that all monomials of maximal degree have positive coefficients. In the proof below, we will only work with symbolic trees that are ultimately positive, i.e.~all polynomials used in them are ultimately positive. For such symbolic trees, all sufficiently large parameters generate positive weighted trees.

Our first observation is that different symbolic trees will generate non-isomorphic trees for many choices of parameters. 

 \begin{lemma}
 \label{lem:symbolic-must-differ-on-parameters} If two symbolic trees $s$ and $t$ over the same variables are not equal, then for every integer $N$ there is a combination $\bar a$ of parameters larger than $N$ such that $s(\bar a) \neq t(\bar a)$.
 \end{lemma}
 \begin{proof}
Let $p_1,\ldots, p_n$ be all the polynomials present in $s$ and $t$, listed without repetition. For any $N$ there is a combination $\bar{a}$ of parameters larger than $N$ such that all values $p_i(\bar a)$ are pairwise different. 
Then $s$ can be reconstructed from $s(\bar a)$ by looking at degrees of all nodes, starting from the leaves, and similarly for $t$ and  $t(\bar a)$.
 \end{proof}
From the well-known Schwartz-Zippel lemma a stronger statement follows: if the parameters $\bar a$ are randomly chosen from $\set{1,\ldots,M}^X$ then the probability that $s(\bar a) \neq t(\bar a)$ approaches $1$ with $M\to\infty$. However, in the following we will only need the crude Lemma~\ref{lem:symbolic-must-differ-on-parameters}.
 
%

 \subsection{Applying patterns to symbolic trees}  
Consider a symbolic tree $t$, and a pattern $\Phi$. If we choose some parameters $\bar a$ sufficiently large for all polynomials in $t$ being positive, and apply $\Phi$ to $t(\bar a)$, then we get some output tree. The following lemma shows that the dependence of the output tree on the parameter $\bar a$ can be described using a symbolic tree. In other words, the lemma gives a way to apply patterns to symbolic trees.

 \begin{lemma}\label{lem:extend-to-polynomial-trees}
 Let $\Phi$ be a pattern and let $t$ be a symbolic tree with variables $X$. There exists a symbolic tree $s$ with the same variables $X$, such that 
 \begin{align*}
 \Phi(t(\bar a)) = s(\bar a)
 \end{align*}
for all sufficiently large parameters $\bar a$.
 \end{lemma}

 \begin{proof} See Appendix~\ref{app:symbolic-pattern-application}. 
 Note that thanks to Lemma~\ref{lem:symbolic-must-differ-on-parameters}, the symbolic tree $s$ is unique.
\end{proof}
 
\begin{myexample}
\label{ex:symb-patterns1}
Consider a simple scenario where the input trees and the patterns are of height one, and there is only one input label (i.e., input trees are effectively unlabelled). For the very simple symbolic input tree $t$ and the pattern $\Phi$ as seen below, the symbolic tree that arises from Lemma~\ref{lem:extend-to-polynomial-trees}:
\mypic{27}
has only one edge labelled with the polynomial
\[
	p(x) = x^4-6x^3+13x^2-7x = x(x-1)(x-2)(x-3) + x(x-1) + x(x-1) + x.
\]
Indeed, for $x>3$, this polynomial counts embeddings of the trees that label the nodes of $\Phi$ into the tree of height one with $x$ children of the root. Note that the polynomial retains full information about the pattern, assuming that both the pattern and all labels of its nodes are of height one. In other words, the single-edge input symbolic tree distinguishes all patterns of this shape. Lemma~\ref{lem:patterns-from-polynomials} below says that this is a general phenomenon, although more complex input trees will be needed to distinguish more complex patterns.
\end{myexample} 

\begin{myexample}
\label{ex:symb-patterns2}
Remaining in the unlabelled setting, consider the following symbolic input tree $t$ and patterns $\Phi_1$ and $\Phi_2$ which are identical except for the embedding between the two bottom nodes:
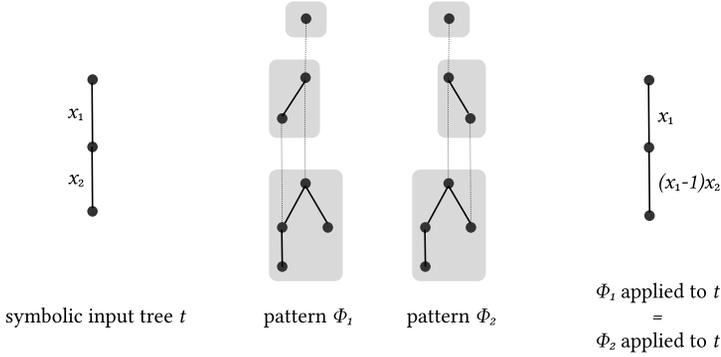
\mypic{28}
The symbolic output trees that arise from Lemma~\ref{lem:extend-to-polynomial-trees} are the same for both patterns. In other words, the symbolic input $t$ does not distinguish the patterns. The slightly larger tree $t'$ below does:
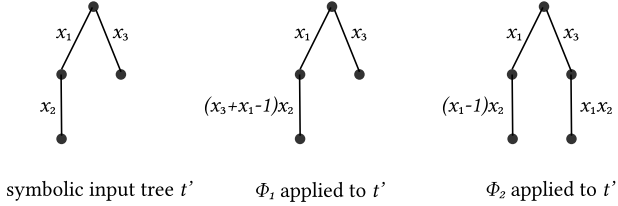
\mypic{29} 
\vspace*{-2ex}
\end{myexample}

\begin{myexample}
\label{ex:symb-patterns3}
Consider a setting with two labels for edges of input trees. For the symbolic input tree $t$ and the pattern $\Phi$ as seen below, the tree that arises from Lemma~\ref{lem:extend-to-polynomial-trees} is shown on the right:
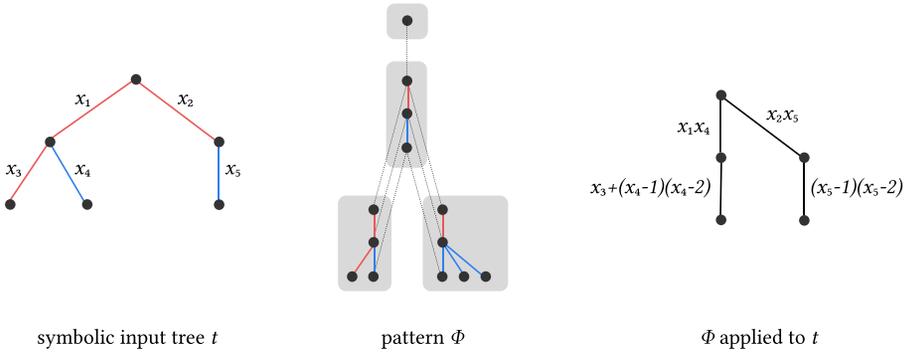
\mypic{25}
Indeed, substituting any positive numbers for the $x_i$'s (with $x_4,x_5>2$), the output tree instance is the result of applying $\Phi$ to the input tree instance.
\end{myexample}

We will only apply the Lemma~\ref{lem:extend-to-polynomial-trees} to symbolic trees which are {\em free}, meaning that all polynomials in them are pairwise distinct variables, as in Examples~\ref{ex:symb-patterns1}-\ref{ex:symb-patterns3}. We will also require the symbolic trees to be sufficiently large for $\Phi$, in the following sense:
\begin{definition}
 A tree $t$ is \emph{large} for a class of trees $S$ if 
 for every tree $s\in S$ and its substructure $s'$ (that is, a ``pruned'' tree that arises from $s$ by removing some subtrees),
 every embedding of $s'$ into $t$ extends to an embedding of $s$ into $t$.
\end{definition}
Every finite set of trees admits some large tree; it is enough to take a complete tree of large height and large degree. 

We now state the main technical result in this section, which is that if we apply a pattern to a symbolic input tree that is free and large, then the output symbolic tree uniquely identifies the pattern. Using the notation from the statement of Lemma~\ref{lem:extend-to-polynomial-trees}, this lemma says that $\Phi$ can be reconstructed from $s$.

\begin{lemma}\label{lem:patterns-from-polynomials}
 Let $\Phi_1$ and $\Phi_2$ be patterns, and let $t$ be a tree that is large for the set of trees that are present as node labels in either of the two patterns. View $t$ as a free symbolic tree, with the variables being the edges of $t$, and consider the two output symbolic trees that arise by applying $\Phi_1$ and $\Phi_2$ to it. If these output symbolic trees are equal, then the patterns $\Phi_1$ and $\Phi_2$ are isomorphic. 
\end{lemma}
\begin{proof}
See Appendix~\ref{app:pattern-extraction}.
\end{proof}

To complete the proof of Theorem~\ref{thm:patterns-must-be-isomorphic}, suppose that $\Phi_1$ and $\Phi_2$ are non-isomorphic patterns. Choose some large finite tree $t$ and apply Lemma~\ref{lem:patterns-from-polynomials}. Since the symbolic trees $\Phi_1(t)$ and $\Phi_2(t)$ are not equal, by Lemma~\ref{lem:symbolic-must-differ-on-parameters}, there must be some choice of parameters $\bar a$ such that applying $\Phi_1$ and $\Phi_2$ to the input tree $t(\bar a)$ yields non-isomorphic outputs, thus proving Theorem~\ref{thm:patterns-must-be-isomorphic}. 

\section{Deciding equivalence for polyregular functions on multiset types}
\label{sec:equivalence-full}
So far we have solved the equivalence problem for quantifier-free tree-to-tree interpretations on trees of bounded height. We shall now generalize that to polyregular functions on arbitrary multiset types. The more important part of the generalization is the ability of polyregular functions to use quantifiers; in this sense, the proof of the theorem below will amount to a quantifier elimination procedure that will reduce the problem to the quantifier-free case from Section~\ref{sec:quantifier-free-tree-to-tree}.

 \begin{theorem}\label{thm:decidable-equivalence}
 The following problem is decidable:
 \begin{description}
 \item[Instance.] Two polyregular functions between multiset types $f,g:\Sigma\to\Gamma$
 given as first-order interpretations, or as derivations.
 \item[Question.] Are the functions equivalent in the following sense: for every input, the two outputs are isomorphic?
 \end{description}
 \end{theorem}


Since the proof of Theorem~\ref{thm:interpretations-equal-derivable} is effective, i.e.~for every first-order interpretation we can compute a corresponding derivation and {\em vice versa}, the representation of the input is not important as long as decidability is concerned. 

We will prove Theorem~\ref{thm:decidable-equivalence} by reducing it to the equivalence problem for quantifier-free tree-to-tree interpretations that was shown decidable in Theorem~\ref{thm:decidable-equivalence-qf-trees}. The reduction is non-elementary, i.e.~both its running time and the size of the output instance are towers of exponentials whose height is polynomial in the input size. 

\subsection{Quantifier elimination}

The most important part of the reduction is a quantifier-elimination result, given in Lemma~\ref{lem:quantifier-elimination} below. The idea is that every structure of a multiset type can be labelled with extra information so that first-order formulas can be replaced by quantifier-free formulas. 

To formalise the notion of labelling, we use \emph{non-copying interpretations}. These are interpretations in which there is only one component, and it has dimension one. For such an interpretation, the universe of the output structure can be regarded as a subset of the universe of the input structure. 

Recall that the \emph{quantifier rank} of a first-order formula is the maximal nesting depth of its quantifiers. For a rank $r \in \set{0,1,\ldots}$, the {\em $r$-theory} of a structure $A$ with distinguished elements $\bar a \in A^n$ is the set of all first-order rank-$r$ formulas with $n$ free variables that are satisfied in it. (This is often called the $r$-type, but we should not overload the word ``type''.) 

\begin{lemma}\label{lem:quantifier-elimination}
For every multiset type $\Sigma$ and numbers $r,n\geq 0$, there is a multiset type $\Gamma$ and a surjective\footnote{Meaning that every structure in $\Sigma$ arises as $f(A)$ for some $A\in\Gamma$, up to isomorphism.}
  non-copying interpretation $f : \Gamma\to \Sigma$,
 such that for every rank-$r$ formula $\varphi(x_1,\ldots,x_n)$ over the vocabulary of $\Sigma$ there is a quantifier-free formula $\psi(x_1,\ldots,x_n)$ 
 over the vocabulary of $\Gamma$ such that for every structure $A \in \Gamma$ and every tuple $\bar a\in A^n$:
 \begin{align*}
 A, \bar a \models \psi \quad \text{iff} \quad \text{all elements of $\bar a$ are in the universe of $f(A)$ and }f(A), \bar a\models \varphi.
 \end{align*}
\end{lemma}

\begin{proof}
See Appendix~\ref{app:quantifier-elimination}.
\end{proof}

\subsection{Reduction to Theorem~\ref{thm:decidable-equivalence-qf-trees}}

We now use quantifier elimination from Lemma~\ref{lem:quantifier-elimination} to reduce the problem in Theorem~\ref{thm:decidable-equivalence}, about equivalence for polyregular functions on multiset types, to the already solved problem from Theorem~\ref{thm:decidable-equivalence-qf-trees}, about equivalence for quantifier-free interpretations on tree types.

Consider an instance 
of the equivalence problem in Theorem~\ref{thm:decidable-equivalence}. Apply Lemma~\ref{lem:embedding} to the output type $\Gamma$, yielding some encoding: 
\[\xymatrix{
 \Sigma \ar@/^/[r]^{f_1} \ar@/_/[r]_{f_2}
 &
 \Gamma
 \ar@{^(->}[r]_-{\text{encode}}
 &
 \multisets^k 1
} \]
The encoding is injective, since it admits an inverse decoding. As a result, appending it to both functions does not change the answer to the equivalence problem. 

Let $r$ be the maximal quantifier rank among all of the first-order formulas used by the above two interpretations. Apply Lemma~\ref{lem:quantifier-elimination} yielding some surjective function $g$ as in the following diagram:
\[\xymatrix{
\Gamma'\ar@{->>}[r]^g 
&
 \Sigma \ar@/^/[r]^{f_1} \ar@/_/[r]_{f_2}
 &
 \Gamma
 \ar@{^(->}[r]_-{\text{encode}}
 &
 \multisets^k 1
} \]
Because $g$ is surjective, prepending it to the diagram does not change the answer to the equivalence problem. 
By the quantifier-elimination properties in Lemma~\ref{lem:quantifier-elimination}, the two functions of type $\Gamma' \to \multisets^k 1$ in the above diagram are quantifier-free interpretations. 

So far, we have managed to reduce the equivalence problem to the special case of quantifier-free interpretations that output unlabelled trees. The last remaining step in the reduction to Theorem~\ref{thm:decidable-equivalence-qf-trees} is turning the input type into trees. This works thanks to the following fact:

\begin{lemma}
 For every multiset type $\Gamma$ there is a finite multiset type $\Delta$, numbers $k,n \in \set{0,1,\ldots}$ and a surjective quantifier-free interpretation 
 \begin{align*}
  \myunderbrace{
    \edgetrees {k} \Delta
  + \cdots + 
    \edgetrees {k} \Delta
  }
  {$n$ {\em times, with trees using edge representation}}
  \to \Gamma.
 \end{align*}
 \end{lemma}
\begin{proof}
Induction on $\Gamma$. For types of the form $\multisets\Sigma$, use the fact that if an interpretation $f$ is surjective and quantifier-free then so is $\multisets f$. 
\end{proof}

\subsection{Connections with the string-to-string case}
\label{sec:connections-with-strings}
In Theorem~\ref{thm:decidable-equivalence}, we proved decidability of equivalence for tree-to-tree interpretations, with the trees being unordered and of bounded height. Decidability of the equivalence problem for string-to-string interpretations remains open. We finish this section with some comments about the connection between these two problems. As we explain below, these two problems are incomparable, because (1) as inputs, strings are more general than trees of bounded height; and
(2) as outputs, trees of bounded height (in fact, height two) are more general than strings.
In particular, a common generalisation of both problems would be interpretations that input strings and output trees of bounded height. Below, we give informal explanations for (1) and (2) above. Consider two interpretations:
  \mypic{30}
  The injective interpretation shows that trees of height two are harder than strings as outputs. The surjective one shows that strings are harder than trees of bounded height as inputs;  the height bound on trees is needed there so that the interpretation can match brackets.

We believe that the ideas developed in this paper, together with above representation of output strings as trees, might be useful in the solution of the string-to-string problem.
\section{Applications to graphs of bounded tree depth}
\label{sec:tree-depth}
\newcommand{\treedepth}{\mathsf G}
In this section, we apply the results on tree-to-tree transducers to graph-to-graph transducers. The graphs in this section are undirected and without parallel edges. 
We use a straightforward representation of a graph as a structure: the universe is the vertices, and there is one binary relation describing the edges. Choosing another representation, such as having edges in the universe, would give the same results. We use \fo interpretations to transform graphs into other graphs.

We would like to decide equivalence for \fo interpretations that describe graph-to-graph functions. 
The problem is undecidable in general, even for interpretations with Boolean outputs. This is because the following \emph{satisfiability problem} is undecidable: given a first-order sentence, decide if it is true in some finite graph~\cite[Theorem 5.5]{courcelleGraphStructureMonadic2012}. A classical approach to working around this limitation is to consider graphs that are similar to trees. Several graph parameters can be used to formalize similarity to trees, with three important parameters being: tree-depth~\cite[Section 6.2]{nevsetvril2012sparsity}, tree-width~\cite[Section 4.1]{courcelleGraphStructureMonadic2012} and clique-width~\cite[Section 4.3]{courcelleGraphStructureMonadic2012}. The parameters are related as follows: every class of graphs satisfies the following implications 
\begin{align*}
\text{bounded tree-depth} 
\quad \Rightarrow \quad 
\text{bounded tree-width} 
\quad \Rightarrow \quad 
\text{bounded clique-width.} 
\end{align*}
For graphs of bounded clique-width (and therefore also for graphs of bounded tree-depth or tree-width), satisfiability is decidable, even if \fo logic is extended to \mso logic~\cite[Theorem 5.80]{courcelleGraphStructureMonadic2012}. It is not known if one can decide equivalence for \fo interpretations between graphs of bounded clique-width. In fact, the problem is open already for tree-width, even for tree-width 1 and dimension 1, see~\cite[p.7]{remarksGraphtoGraph}. We present the first progress on this problem: we show that the equivalence problem is decidable for bounded tree-depth, for interpretations of arbitrary dimension. We begin by defining tree-depth. 
\begin{definition}
 A graph has tree-depth $k=1$ if it has no edges. 
 A graph has tree-depth at most $k > 1$ if there is a vertex, such that after removing that vertex, all connected components have tree depth at most $k-1$. 
\end{definition}



\begin{theorem}\label{thm:decidable-equivalence-tree-depth}
 The following problem is decidable:
 \begin{description}
 \item[Instance.] Two \fo interpretations, which input graphs of tree-depth at most $k$ and output graphs of tree-depth at most $\ell$. 
 \item[Question.] Are the interpretations equivalent in the following sense: for every input, the two output graphs are isomorphic?
 \end{description}
 \end{theorem}

 Before proving the theorem, let us make two remarks. First, the above problem would remain decidable if \mso logic was used instead of \fo logic. This is because, as we have mentioned before, \fo and \mso have the same expressive power on graphs of bounded tree-depth~\cite[Theorem 1.1]{elberfeld2016first}. The second remark is about how the instances of the decision problem are represented. We use the following representation: we are given the bounds $k$ and $\ell$ on the tree-depth, as well as a first-order interpretation that uses the vocabulary of graphs. The interpretation should have the property that if the input graph has tree-depth at most $k$, then the output graph has tree-depth at most $\ell$. This property is decidable, since \fo logic is decidable on graphs of tree-depth at most $k$, and there is a \fo formula that checks if a graph has tree-depth at most $\ell$. 

 We now proceed to prove Theorem~\ref{thm:decidable-equivalence-tree-depth}.
The difficulty in the theorem is that there is an implicit isomorphism check, i.e.~the two interpretations might produce isomorphic graphs in different ways. This difficulty was already present in Theorems~\ref{thm:decidable-equivalence-qf-trees} and~\ref{thm:decidable-equivalence}, because bounded-height trees also have non-trivial isomorphisms. As it turns out, that difficulty has already been addressed: Theorem~\ref{thm:decidable-equivalence-tree-depth} is proved by a relatively straightforward reduction to the case of multiset types from Theorem~\ref{thm:decidable-equivalence}. The key step is to represent graphs of bounded tree-depth using multiset types:

\begin{lemma}\label{lem:inj-surj}
 For every $k \in \set{1,2,\ldots}$ there are:
 \begin{enumerate}
    \item a surjective interpretation from some multiset type to graphs of tree-depth at most $k$;
    \item an injective interpretation from graph of tree-depth at most $k$ to some multiset type.
 \end{enumerate}
\end{lemma}
  
Before proving the lemma, we use it to complete the proof of Theorem~\ref{thm:decidable-equivalence-tree-depth}. Consider an instance of the problem in the theorem. Apply Lemma~\ref{lem:inj-surj}, yielding a surjective interpretation onto the input graphs, and an injective interpretation from the output graphs, as shown in the following diagram:
\[
 \begin{tikzcd}
 \Sigma 
 \ar[r,two heads]
 &
\treedepth_k
 \ar[r,bend left,"f_1"]
 \ar[r,bend right,"f_2"']
 &
\treedepth_\ell
 \ar[r,hookrightarrow]
 &
 \Gamma
 \end{tikzcd}
 \]
 where $\treedepth_k$ (and $\treedepth_\ell$) denote the class of graphs of tree-depth at most $k$ (and $\ell$).
The equivalence of $f_1$ and $f_2$ is the same as the equivalence of the two paths in the diagram that go from $\Sigma$ to $\Gamma$, and these paths represent interpretations between multiset types. For such interpretations, equivalence is decidable by Theorem~\ref{thm:decidable-equivalence}. It remains to prove the lemma. 

\begin{proof}[Proof of Lemma~\ref{lem:inj-surj}]
 In the proof, it will be convenient to work with vertex-labelled graphs. For a finite set $\Sigma$, let us write $\treedepth_k \Sigma$ for the class of graphs that have tree-depth at most $k$, with vertices labelled by elements of $\Sigma$. Such graphs are viewed as structures in the same way as unlabelled graphs, with an additional unary predicate $a(x)$ for each possible label $a \in \Sigma$. Ultimately we care about unlabelled graphs, i.e.~the case of $\Sigma=1$, but the induction proofs will use other choices of $\Sigma$.
 
 \paragraph*{Surjective.} 
 We begin describing a surjective interpretation from some multiset type to $\treedepth_k \Sigma$, for every finite set $\Sigma$. The interpretation is defined by induction on $k$. For $k=1$, there is nothing to do, since a graph of tree-depth $1$ is nothing but a multiset of vertices, and therefore $\treedepth_1 \Sigma = \multisets \Sigma$. 
 
 Consider now the induction step, where we  already have a surjective interpretation for tree-depth $k$, and we now want to define it for $k+1$. 
 By definition, a graph of tree-depth at most $k+1$ has some distinguished vertex $v$, such that after removing this vertex, every connected component has strictly smaller tree-depth. Therefore, in order to represent such a graph, it is enough to give: (a) the label of the removed vertex; (b) the multiset of graphs that represent the connected components after removing the vertex; and (c) for each vertex that is not removed, one bit of information that says if it was connected to the removed vertex by an edge. This is formalized in the following claim:

 \begin{claim}\label{claim:surjective-rep}
 For every finite set $\Sigma$ of labels, there is a surjective interpretation 
 \[
 \begin{tikzcd}
\Sigma\times \quad \multisets ( \treedepth_{k}(\Sigma \times \bool))
 \qquad 
 \ar[r,two heads]
 & 
 \qquad 
 \treedepth_{k+1} \Sigma.
 \end{tikzcd}
 \]
 \end{claim}
 \begin{proof}
 The interpretation maps an element of the input type to a graph in the natural way: it takes the disjoint union of the graphs in the multiset from the second coordinate, adds a new vertex with the label from the first coordinate, connects this new vertex to the vertices in the disjoint union that have a $1$ value on the extra bit, and removes this extra bit. This function is surjective by definition of tree-depth, and it is easily seen to be a first-order interpretation. 
 \end{proof}

Combining the above claim with an inductively defined interpretation from some multiset type to $\treedepth_{k}(\Sigma \times \set{0,1})$, we get the desired surjective interpretation from a multiset type to $\treedepth_{k+1} \Sigma$. 

 \paragraph*{Injective.} We now present an injective interpretation from graphs of tree-depth at most $k$ to a multiset type. The rough idea is to use the inverse of the surjective interpretation described above. However, the rough idea needs some work, since there can be several choices for the removed vertex, and there might be no way of choosing a unique removed vertex. The solution is to output all possible choices, aggregated using a multiset. This is expressed in the following claim. 
 
 \begin{claim}\label{claim:representation-inj}
 For every finite set $\Sigma$ of labels, there is an injective interpretation 
 \[
 \begin{tikzcd}
 \treedepth_{k+1} \Sigma
 \ar[r,hookrightarrow]
 \qquad 
 & 
 \qquad 
 \multisets(
 \Sigma \times \multisets ( \treedepth_{k}(\Sigma \times \bool)))
 \end{tikzcd}
 \]
 \end{claim}
 \begin{proof} 
 This interpretation is defined as follows. Consider a vertex $v$ of an input graph $G$, and let $H$ be a connected component of the graph obtained by deleting $v$ in $G$. Similarly to Claim~\ref{claim:surjective-rep}, define $H^v$ to be the graph with labels from $\Sigma \times \bool$ that is obtained from $H$ by adding to the label of each vertex one bit of information that says if that vertex is connected to the deleted vertex $v$ by an edge. The interpretation in the claim maps a graph $G$ to 
 \begin{align*}
 \left\{\!\!\!\left\{\left(\text{label of $v$},
 \left\{\!\!\!\left\{
 H^v
 \middle|
 \begin{tabular}{l}
 $H$ is a connected \\ 
 component of $G$ with \\
 vertex $v$ removed
 \end{tabular}
 \right\}\!\!\!\right\}
 \right)\,\,
 \middle|
 \begin{tabular}{l}
 $v$ is a vertex of $G$ such that after \\
 removing it, every
 connected \\ 
 component  has tree-depth at most $k$
 \end{tabular}
 \right\}\!\!\!\right\}
 \end{align*}

 This interpretation is injective, because it inverses the function from Claim~\ref{claim:surjective-rep} in the following sense: if we apply the function from this claim to a graph in $\treedepth_{k+1} \Sigma$, and then apply the interpretation from Claim~\ref{claim:surjective-rep} to each element of the resulting multiset, then we get a multiset that contains several copies of the original input graph. 

 It remains to prove that the function described above is indeed  a \fo interpretation. The main observation is that the criterion for $v$ in the outermost multiset is definable in \fo logic assuming that the input graph has bounded tree-depth. The reason is that, although \fo logic cannot define graph connectivity in general, it can do so for graphs of bounded tree-depth, because in graphs of tree-depth $k$, paths have length at most $2^k$. Therefore, the operation described in this claim can be defined using a \fo interpretation. This interpretation has dimension two, because elements of the output multiset of multisets are represented using pairs $(v,w)$ such that $v$ is a deleted vertex and $w$ is one of the remaining vertices. 
 \end{proof}
Similarly to the surjective case, the injective case is proved by induction on $k$ using the above claim in the induction step. This completes the proof of Lemma~\ref{lem:inj-surj}. 
\end{proof}

\begin{acks}
The idea to consider regular functions on multisets was suggested to us by Marcelo Fiore. We are very grateful to the anonymous reviewers for their numerous insightful comments and corrections.

Miko{\l}aj Boja\'nczyk was supported by the \grantsponsor{grant:ncn}{Polish National Science Centre (NCN)}{} grant ``Polynomial finite state computation’' (\grantnum{grant:ncn}{2022/46/A/ST6/00072}).
\end{acks}

\bibliographystyle{plainnat}
\bibliography{bib}

\newpage
\appendix

\section*{Appendix}

\section{Proof of Lemma~\ref{lem:extend-to-polynomial-trees}}
\label{app:symbolic-pattern-application}

 \begin{proof}
Fix a symbolic tree $t$ with variables $X$, and a pattern $\Phi$.

In the proof, we use tree homomorphisms. A homomorphism is defined like an embedding, except that it need not be injective. In other words, a homomorphism maps edges in the input tree to edges in the output tree that have the same colour, with the root being mapped to the root. These are exactly the homomorphisms for trees regarded as relational structures under the edge representation introduced early in Section~\ref{sec:quantifier-free-tree-to-tree}.

When talking about a homomorphism into a symbolic tree, we only care about the colours in the symbolic tree and we ignore the polynomials. Here is an example.
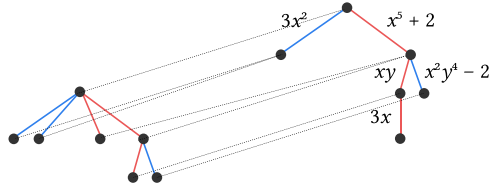
\mypic{21}
The key notion in the proof of Lemma~\ref{lem:extend-to-polynomial-trees}, also used in the proof of Lemma~\ref{lem:patterns-from-polynomials} later, is the notion of a profile. 

\begin{definition}
 A \emph{profile} for a pattern $\Phi$ and a symbolic tree $t$ is a pair consisting of a node $x$ in the pattern and a homomorphism $ h : \Phi(x) \to t$. 
\end{definition}

Once we fix $\Phi$ and $t$, profiles can be organized into a tree, called the \emph{tree of profiles}, with the children of a profile $(x_1,h_1)$ being all profiles $(x_2,h_2)$ such that $x_2$ is a child of $x_1$ in $\Phi$ and $h_2$ extends $h_1$.
The tree of profiles has (at most) the same height as $\Phi$ but can be much wider than $\Phi$, since for each node $x$ there might be multiple possible homomorphisms $h$. However, the root of the tree of profiles is unique, since the root of $\Phi$ is labelled by the one-node tree which admits a unique homomorphism to every tree.

If we apply $\Phi$ to a positive instantiation $t(\bar a)$ of the symbolic tree $t$, then each node of the output tree will have a corresponding profile, as described below. 
 Consider a choice of parameters $\bar a \in \Nat^X$. 
 For each node in the input tree $t (\bar a)$ there is a corresponding node in the symbolic tree $t$ that was used to generate it, which is called its \emph{origin}. The origin function
 \begin{equation}
 \label{eq:origin-map}
 \xymatrix{
 t(\bar a) 
 \ar[r]^{\text{origin}}
 & t,
 }
 \end{equation}
 is a tree homomorphism (but not an embedding, since several nodes from $t(\bar a)$ might have the same origin).
 Note that if two nodes have the same origin then they have isomorphic subtrees in $t(\bar a)$; moreover, there is an automorphism of $t(\bar a)$ that maps one of the nodes to the other.
 
 Let us now look at nodes of the output tree $\Phi(t(\bar a))$. By the semantics of patterns, each such node is created by taking a node $x$ in the pattern and some embedding
 \begin{equation}
 \label{eq:node-embedding}
 \alpha : 
 \Phi(x) \hookrightarrow t (\bar a).
 \end{equation}
 Define the \emph{profile} of such a node in the output tree to be the node $x$ in the pattern together with the homomorphism that is obtained by composing the tree embedding~\eqref{eq:node-embedding} with the origin homomorphism~\eqref{eq:origin-map}. The crucial property of profiles is that they determine subtrees of the output tree.

 \begin{claim}
 If two nodes of the output tree $\Phi(t(\bar a))$ have the same profile, then they have isomorphic subtrees in the output tree.
 \end{claim}
 \begin{proof}
 Suppose that both profiles use the same node $x$, and the corresponding embeddings are 
\begin{align*}
\alpha_1,\alpha_2 : \Phi(x) \hookrightarrow t(\bar a).
\end{align*}
Since the profiles of the two embeddings are the same, $\alpha_1$ and $\alpha_2$ are equalised by the origin homomorphism~\eqref{eq:origin-map}. As a result, there is an automorphism $\pi$ of the input tree $t(\bar a)$ that makes the following diagram commute: 
 
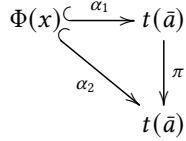
\[\xymatrix{
 \Phi(x) 
 \ar@{^(->}[r]^{\alpha_1}
 \ar@{^(->}[dr]_{\alpha_2}
 &
 t(\bar a)
 \ar[d]^{\pi}
 \\
 & 
 t(\bar a )
}\]
 This automorphism witnesses isomorphism of the trees in the statement of the claim.
 \end{proof}
 
Thanks to the above claim, for every profile and every choice of parameters $\bar a$ there is a unique tree that arises as a subtree of the output tree that is rooted in a node with the given profile. We will show that this unique tree is in fact obtained by instantiating, with parameters $\bar a$, a certain symbolic tree that does not depend on $\bar a$. This will complete the proof of Lemma~\ref{lem:extend-to-polynomial-trees}, since the entire output tree is obtained from the symbolic tree arising from the root profile.

 The symbolic tree corresponding to a profile is defined by induction on the tree of profiles, starting with the leaves. If a profile is a leaf in the tree of profiles, then the corresponding subtree of the output tree has only one node, and no polynomials are used (since there are no edges).

 Consider now a profile $(x,h)$ that is not a leaf in the tree of profiles. The symbolic tree for the profile $(x,h)$ is defined as follows: for every child $(y,g)$ we take the inductively defined symbolic tree, and connect it to the root by an edge that is labelled by a certain polynomial that counts embeddings:

 \begin{claim}\label{claim:edge-polynomial}
 Let $(y,g)$ be a profile that has a defined parent profile $(x,h)$. There is a polynomial $p \in \Int[X]$ such that for all sufficiently large parameters $\bar a \in \Nat^X$, $p(\bar a)$ is the number of ways in which an embedding $\alpha:\Phi(x)\hookrightarrow t(\bar a)$ with profile $(x,h)$ can be extended to an embedding $\beta:\Phi(y)\hookrightarrow t(\bar a)$ with profile $(y,g)$. 
 \end{claim}
 \begin{proof}
 The following diagram illustrates the various tree embeddings and homomorphisms involved:
 
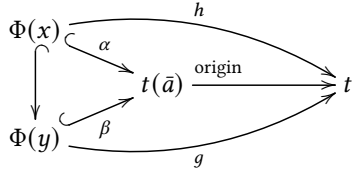
\[\xymatrix@R=5pt{
 \Phi(x)\ar@{^(->}[dd]\ar@{^(->}[rd]^\alpha\ar@/^3ex/[rrrd]^h \\
 & t(\bar a)\ar[rr]^(.3){\text{origin}} && t \\
 \Phi(y)\ar@{^(->}[ru]_\beta\ar@/_3ex/[rrru]_g
 }\]
 For every edge $e$ of the symbolic tree $t$, define 
 \begin{align*}
 k_e \quad & \text{the number of edges mapped to $e$ by the homomorphism $h$}\\
 \ell_e \quad & \text{the number of edges mapped to $e$ by the homomorphism $g$}
 \end{align*}
 Note that $k_e\leq\ell_e$. Then put: 
 \begin{align*}
p(\bar a)  =  \prod_{e \in \text{edges of $t$}} \prod_{i=k_e}^{\ell_e-1}((\text{polynomial labelling the edge $e$ in $t$})-i).
 \end{align*}
 \end{proof}

 Using the polynomial from the above claim, it is easy to define the symbolic tree for a profile $(x,h)$; it is the linear combination 
 \begin{align*}
 \sum_{(y,g)} \text{(polynomial from Claim~\ref{claim:edge-polynomial})} \cdot \text{(symbolic tree for $(y,g)$)}.
 \end{align*}

 This completes the proof of Lemma~\ref{lem:extend-to-polynomial-trees}. 
 \end{proof}
 
\section{Proof of Lemma~\ref{lem:patterns-from-polynomials}}
\label{app:pattern-extraction}

In the proof, it will be more convenient to work with a single pattern, and to analyze the different profiles for that pattern. We can merge any two patterns $\Phi_1$ and $\Phi_2$ into a single pattern $\Phi$, by joining them with a common root that is labelled by a singleton tree. In the pattern $\Phi$, the root has two children $x_1$ and $x_2$ that correspond to the patterns $\Phi_1$ and $\Phi_2$. Since each of the nodes $x_i$ is labelled by a singleton tree, it admits a unique homomorphism $h_i$ to the symbolic tree $t$. Consider two profiles $(x_1,h_1)$ and $(x_2,h_2)$ corresponding to these nodes. 
Recall that in the proof of Lemma~\ref{lem:extend-to-polynomial-trees}, we showed that for every profile there is a unique symbolic tree that describes the corresponding subtree in the output tree; we call this tree \emph{the symbolic tree} of the profile. The symbolic trees of the profiles $(x_1,h_1)$ and $(x_2,h_2)$ are the same, respectively, as the symbolic trees of the patterns $\Phi_1$ and $\Phi_2$. At the end of this proof, we will show in Claim~\ref{claim:deep-isomorphism} that if the symbolic trees of the profiles $(x_1,h_1)$ and $(x_2,h_2)$ are equal, then the patterns $\Phi_1$ and $\Phi_2$ are isomorphic, thus completing the proof of the lemma.

\paragraph*{Recovering a tree from a polynomial.} The general idea is that the structure of a profile and its descendants can be determined by a syntactic analysis of the polynomials that appear in its symbolic tree. We begin some simple observations used in this syntactic analysis. 
For a tree homomorphism $h$, define its \emph{multiset} $m(h)$ to be the multiset of edges in the target tree which says how many times each edge of the target tree is produced by some edge in the source tree. Consider the symbolic tree of a profile $(x,h)$ in some pattern $\Phi$, for some symbolic input tree $t$ which is free. Like any symbolic tree, this symbolic tree is a linear combination 
\begin{align*}
\sum_{s \in S} p_s \cdot s,
\end{align*}
where $S$ is a finite set of smaller symbolic trees and each coefficient $p_s$ is a polynomial that uses the same variables as the symbolic input tree; these variables are the edges of the tree $t$. Take some $s \in S$. The polynomial $p_s$ assigns to each monomial some coefficient. Since monomials can be viewed as multisets of variables, and the variables in this case are the edges of the tree $t$, we can view the polynomial $p_s$ as a function from multisets of edges in $t$ to integers. Of course, only finitely many multisets are assigned nonzero coefficients. Furthermore, since the polynomial $p_s$ is ultimately positive, the monomials of maximal degree (which correspond to inclusion-wise maximal multisets of edges in $t$) have positive coefficients. The following claim shows that the coefficients for monomials of maximal degree reveal certain information about the children of a profile. 
\begin{claim}\label{claim:parse-coefficients}
 Let $(x,h)$ be a profile for some pattern $\Phi$ and some free symbolic tree $t$, with the symbolic tree of the profile equal to some
 \begin{align*}
 \sum_{s \in S} p_s \cdot s.
 \end{align*}
 Let $s \in S$ and let $Y$ be a multiset of edges in the tree $t$. In the polynomial $p_s$, the 
 coefficient of (the monomial corresponding to) the multiset of edges $Y$ is equal to the number of profiles $(y,g)$, which: 
 \begin{enumerate}
 \item are children of $(x,h)$ in the profile tree;
 \item have the symbolic tree equal to $s$;
 \item satisfy $Y = m(g) - m(h)$.
 \end{enumerate}
\end{claim}
\begin{proof}
Follows from the construction used in the proof of Claim~\ref{claim:edge-polynomial}, for the special case of $t$ being free.
\end{proof}

\paragraph*{Isomorphisms between profiles.} Using Claim~\ref{claim:parse-coefficients} we will recover a profile from its symbolic output tree. This will only be possible up to an isomorphism of profiles, and we begin by making precise the notion of profile isomorphism.

 For two profiles $(x_1,h_1)$ and $(x_2,h_2)$, a \emph{shallow isomorphism} is defined to be an isomorphism between the two trees $\Phi(x_1)$ and $\Phi(x_2)$ that is consistent with the homomorphisms in the sense that the following diagram commutes:
\[\xymatrix{
\Phi(x_1)
\ar[d]_{\text{shallow isomorphism}}
\ar[r]^{h_1} 
& t\\
\Phi(x_2)
\ar[ur]_{h_2}
}\]
If the profile is injective, which means that its homomorphism is injective, then its shallow isomorphism class is determined by the image of its homomorphism, which is an induced subtree of $t$. This image is stored in the mulitiset $m(h)$, which is non-repeating in the case when $h$ is injective. This observation is summarised in the following claim.
\begin{claim}\label{claim:shallow-isomorphism}
 For two injective profiles $(x_1,h_1)$ and $(x_2,h_2)$, there is a shallow homomorphism if and only if $m(h_1)=m(h_2)$. 
\end{claim}

What we actually care about is {deep isomorphism} between profiles, which corresponds to isomorphism of entire patterns. A \emph{deep isomorphism} between two profiles $(x_1,h_1)$ and $(x_2,h_2)$ is defined to be a bijection $f$ between the subtrees of these profiles in the tree of profiles, together with a family of shallow isomorphisms 
\[\xymatrix{
\{(y_1,g_1) 
\ar[r]^-{f_{(y_1,g_1)}}
&
f(y_1,g_1)\}_{(y_1,g_1)},
}\]
indexed by profiles in the subtree of $(x_1,h_1)$, such that the shallow isomorphisms in the family are consistent with the extension structure of profiles. Equipped with this notion of deep isomorphism, we are ready to state the main claim in the proof of Lemma~\ref{lem:patterns-from-polynomials}.

\begin{claim}\label{claim:deep-isomorphism}
 Let $\Phi$ be a pattern, and let $t$ be a free symbolic tree that is large for all trees that appear as node labels in $\Phi$. 
 Let $(x_1,h_1)$ and $(x_2,h_2)$ be two injective profiles that admit a shallow isomorphism. Then the following conditions are equivalent: 
 \begin{enumerate}
 \item there is a deep isomorphism between $(x_1,h_1)$ and $(x_2,h_2)$;
 \item the symbolic trees for $(x_1,h_1)$ and $(x_2,h_2)$ are equal.
 \end{enumerate}
\end{claim}

Before proving the claim, we use it to complete the proof of Lemma~\ref{lem:patterns-from-polynomials}. Consider two patterns $\Phi_1$ and $\Phi_2$ as in the assumption of the lemma, and consider the pattern $\Phi$ constructed by taking the union of two patterns $\Phi_1$ and $\Phi_2$ joined by a common root, as described before. The profiles $(x_1,h_1)$ and $(x_2,h_2)$ that correspond to the roots of $\Phi_1$ and $\Phi_2$ will admit a deep isomorphism if and only if the patterns are isomorphic. Also, the profiles are trivially injective, since the labels of the nodes $x_1$ and $x_2$ are singleton trees; for the same reasons the profiles admit a shallow isomorphism. The symbolic trees of $\Phi_1$ and $\Phi_2$ are the same as the symbolic trees of $(x_1,h_1)$ and $(x_2,h_2)$. If these symbolic trees are equal, then we can use Claim~\ref{claim:deep-isomorphism} to conclude that the patterns are isomorphic. 

It remains to prove Claim~\ref{claim:deep-isomorphism}.

\begin{proof}
 The implication 1 $\Rightarrow$ 2 is easy, so we only prove the converse implication (which is the one we need anyway). The proof is by induction on two parameters, ordered lexicographically:
 \begin{enumerate}
 \item the size of the pattern $\Phi$;
 \item the number of profiles that are descendants of $(x_1,h_1)$. 
 \end{enumerate}

 The induction basis is when both parameters are one, in which case the claim is trivial, since there is only one profile in the pattern.

 Consider now the induction step. 
 Suppose that two profiles $(x_1,h_1)$ and $(x_2,h_2)$ have the same symbolic tree, which we view as a linear combination of smaller symbolic trees:
\begin{align*}
\sum_{s \in S} p_s \cdot s.
\end{align*}
The smaller symbolic trees $s \in S$ correspond to children of the profiles. If the linear combination is zero, then there are no children for the profiles, and shallow isomorphism is the same as deep isomorphism. 

Suppose now that the linear combination is nonzero, which means that the profiles are not leaves in the tree of profiles. Choose some child $y_1$ of $x_1$ in the pattern $\Phi$ which maximises the number of nodes in the tree $\Phi(y_1)$. Because the symbolic input tree is large, we know that the embedding $h_1$ can be extended to an embedding $g_1$ that yields an injective profile $(y_1,g_1)$ which is a child of $(x_1,h_1)$ in the profile tree. Let $s$ be the symbolic tree that corresponds to the profile $(y_1,g_1)$, and consider the polynomial $p_s$ corresponding to this symbolic tree. 
By maximality of $y_1$, we can apply Claim~\ref{claim:parse-coefficients} to conclude that, in the polynomial $p_s$, the 
 coefficient next to the monomial (corresponding to the multiset $Y = m(g_1) - m(h_1)$)
is equal to the number of children of $(x_1,h_1)$ in the profile tree that satisfy the three conditions in the claim. The same coefficients and polynomials are used for $(x_2,h_2)$, and therefore the same number is appropriate for $(x_2,h_2)$, in particular there must be some child profile $(y_2,g_2)$ which has the same symbolic tree $s$ as $(y_1,g_1)$ and such that the following multiset equality holds:
\begin{align}\label{eq:mulitset-differences}
m(g_1) - m(h_1) = m(g_2) - m(h_2).
\end{align}
By Claim~\ref{claim:shallow-isomorphism} and the assumption that $(x_1,h_1)$ admits a shallow isomorphism with $(x_2,h_2)$, we know that $m(h_1)=m(h_2)$. Again by Claim~\ref{claim:shallow-isomorphism} and~\eqref{eq:mulitset-differences}, we conclude that $(y_1,g_1)$ admits a shallow isomorphism with $(y_2,g_2)$. Since the corresponding symbolic trees are also equal, we can apply the inductive assumption to conclude that the children $(y_1,g_1)$ and $(y_2,g_2)$ admit a deep isomorphism.

So far we have succeeded in creating a deep isomorphism for one specific pair of children. To extend this deep isomorphism to the remaining children, we proceed by induction on the size of the pattern as follows. Consider the node $y_1$ in the pattern $\Phi$. Consider the linear combination 
\begin{align}\label{eq:twins-of-the-maximal-one}
\sum_g\text{symbolic tree of $(y_1,g)$}
\end{align}
where the sum ranges over all profiles $g$ in $\Phi$ that use the node $y_1$. This linear combination is uniquely determined by the deep isomorphism class of the profile $(y_1,g)$, and therefore it is determined by the information that we have already. 
If we subtract the linear combination~\eqref{eq:twins-of-the-maximal-one} from the symbolic tree of the profile $(x_1,h_1)$, then we get the symbolic tree of the profile $(x_1,h_1)$, but in a smaller pattern, call it $\Phi - y_1$, which is obtained from $\Phi$ be removing node $y_1$ and its entire subtree. From the inductive hypothesis applied to the smaller pattern we conclude that the patterns $(x_1,h_1)$ and $(x_2,h_2)$ are deeply isomorphic in the smaller pattern $\Phi - y_1$. Combining this with the deep isomorphism between $(y_1,g_1)$ and $(y_2,g_2)$, we get a deep isomorphism between $(x_1,h_1)$ and $(x_2,h_2)$ in the original pattern $\Phi$, thus completing the proof of Claim~\ref{claim:deep-isomorphism}, and of Lemma~\ref{lem:patterns-from-polynomials}.
\end{proof}

\section{Proof of Lemma~\ref{lem:quantifier-elimination}}
\label{app:quantifier-elimination}

 The proof is by induction on the structure of $\Sigma$. 
 In the base case of $\Sigma =1$ there is nothing to do, and the case of coproducts $\Sigma_1 + \Sigma_2$ is also straightforward. We deal with products and multisets below.

\paragraph*{Products.}
Consider a product type $\Sigma_1 \times \Sigma_2$. Apply the induction assumption to $\Sigma_1$ and $\Sigma_2$, and the same quantifier rank and number of free variables, yielding surjective non-copying interpretations 
 \begin{align*}
f_1:\Gamma_1\to\Sigma_1 \qquad\text{and}\qquad f_2:\Gamma_2\to\Sigma_2.
 \end{align*}
 Combine these functions into: 
 \begin{align*}
 f_1\times f_2 : \Gamma_1 \times \Gamma_2 \to \Sigma_1 \times \Sigma_2.
 \end{align*}
 This function is surjective and non-copying. It remains to prove that the $r$-theories of the output are determined by the quantifier-free theories of the input. This follows from the induction assumption, by the compositionality property of $r$-theories: the $r$-theory of a pair structure is uniquely determined by the $r$-theories of the two structures in the pair. This is easily proved by analysing $r$-round Ehrenfeucht-Fra\"iss\'e games on pair structures, using the fact that in a pair structure $(A,B)$ no relations connect elements of $A$ with elements of $B$.

 
\paragraph*{Multisets.} For a type $\multisets \Sigma$, we also use a compositionality argument. Consider the set of all possible $r$-theories of structures in $\Sigma$ without distinguished elements:
 \begin{align*}
 I = \setbuild{\text{$r$-theory of $A$}}{$A \in \Sigma$}.
 \end{align*}
By properties of \fo logic over finite relational vocabularies, $I$ is a finite set for every $\Sigma$ and $r$. Now consider the following claim, again proved by analysis of $r$-round Ehrenfeucht-Fra\"iss\'e games:
 \begin{claim}\label{claim:multiset-types}
For every structure $A \in \multisets \Sigma$ with tuple of distinguished elements $\bar a \in A^n$, its $r$-theory is determined uniquely by the following information: 
 \begin{enumerate}
 \item The function 
 $\alpha : 
 I \to \set{0,1,\ldots,r}
 $
 which counts for each $i \in I$ how many structures $B \in \Sigma$ in $A$ have this $r$-theory, with the number cut off at threshold $r$. This function does not depend on $\bar a$.
 \item For every subset $X \subseteq \set{1,\ldots,k}$ of coordinates in the tuple $\bar a$, the following information: (a) whether the tuple $\bar a |_X$ is entirely contained in a single structure $B$ in $A$, and (b) if yes then what is the $r$-theory of $B$ with distinguished elements $\bar a|_X$.
 \end{enumerate}
 \end{claim}
 For the following it is not important that the threshold in (1) is $r$; any finite number would suffice. However, $r$ works for the Claim to hold. Indeed, if a multiset contains many elements with the same $r$-theory, we may just as well consider a sub-multiset with only $r$ of these elements, because a formula with at most $r$ nested quantifiers will not be able to make any use of more than $r$ of them.
 
 Now apply the induction assumption from the lemma to the multiset type $\Sigma$, yielding a function $f : \Gamma \to \Sigma$. We assume that the domain $\Gamma$ is partitioned along possible $r$-theories, i.e. 
 \begin{align*}
 \textstyle\Gamma = \coprod_{i \in I} \Gamma_i
 \end{align*}
 and the function $f$ restricted to $\Gamma_i$ produces only structures with $r$-theory $i$. This stronger assumption can be pulled through the induction in the lemma; we gloss over it to keep notation light. 
 
Define the \emph{profile} of a multiset $A \in \multisets \Sigma$ to be the function $\alpha$ from Claim~\ref{claim:multiset-types}. For a profile $\alpha : I \to \set{0,1,\ldots,r}$, define $\Gamma_\alpha$ to be the type 
 \begin{align}\label{eq:hugeproduct}
\textstyle \prod_{i \in I} (\Gamma_i)^{\alpha(i)} \quad \times \quad \prod_{i \in I \text{ s.t. } \alpha(i)= r} \multisets \Gamma_i,
 \end{align}
 and define 
 $
 g_\alpha : \Gamma_\alpha \to \multisets \Sigma
 $
 to be the function that applies $f$ to all elements of $\Gamma$ that appear in its input, and then puts them together into a common multiset. This is a non-copying interpretation. By construction, all outputs of $g_\alpha$ have profile $\alpha$; furthermore, thanks to the presence of the second component in~\eqref{eq:hugeproduct}, $g_\alpha$ is surjective in the sense that every multiset in $\multisets \Sigma$ of profile $\alpha$ is in the image of this function. By (both parts of) Claim~\ref{claim:multiset-types} and the induction assumption, we know that for every $A \in \Gamma_\alpha$ and every tuple of $k$ distinguished elements $\bar a \in A^k$, the $r$-theory of $g_\alpha(A, \bar a)$ is determined uniquely by the quantifier-free theory of $A, \bar a$. 

 To conclude the lemma, consider the interpretation
 \begin{align*}
\textstyle g : (\coprod_\alpha \Gamma_\alpha) \to M\Sigma,
 \end{align*}
 obtained by combining the functions $g_\alpha$ for all the finitely many possible profiles. 

\label{app:quantifier-elimination}

\end{document}